\documentclass[12pt]{article}
\usepackage{amsfonts,amsmath,amssymb,amsthm}
\usepackage{graphicx} 
\usepackage{float}
\usepackage{caption}
\usepackage{subcaption}
\usepackage{graphicx,comment}

\usepackage{fullpage}
\usepackage[margin=1in]{geometry}
\usepackage{xcolor}
\usepackage{hyperref}
\usepackage[capitalize]{cleveref}
\usepackage{authblk}
\usepackage{xurl}

\usepackage[numbers,sort&compress]{natbib}

\newtheorem{defn}{Definition}

\newtheorem*{rem}{Remark}
\newtheorem{lemma}{Lemma}

\newenvironment{myproof}{
  \par\medskip\noindent
  \textit{Proof}.
}{
\newline
\rightline{$\qedsymbol$}
}

\def\eps{\varepsilon}

\def\Pois{\operatorname{Pois}}
\def\sign{\text{sign}}

\newtheorem{theorem}{Theorem}

\def\fp{\rho} %
\def\fps{\rho^R} %

\def\ep{\psi} %
\def\eps{\psi^R} %

\def\p{p}
\def\q{q}
\def\bias{\gamma}
\def\ps{\p^R}
\def\qs{\q^R}
\def\biass{\bias^R}

\def\SI{Appendix}

\title{The evolutionary advantage of replacers in the Moran process
}
\def\maintitle{The evolutionary advantage of replacers in the Moran process
}

\author[1]{Michal Pecho}
\author[1]{Josef Tkadlec}
\author[2,3]{Martin A. Nowak}

\affil[1]{Computer Science Institute, Charles University, Prague, Czech Republic}

\affil[2]{Department of Organismic and Evolutionary Biology, Harvard University, Cambridge, MA02138, USA}
\affil[3]{Department of Mathematics, Harvard University, Cambridge, MA02138, USA}

\date{}

\usepackage{lineno}
\begin{document}

\maketitle

\begin{abstract}
Evolution occurs in populations of reproducing individuals. 
In stochastic descriptions of evolutionary dynamics, such as the Moran process, individuals are chosen randomly for birth and for death.
If the same type is chosen for both steps, then the reproductive event is wasted, because the composition of the population remains unchanged.
Here we introduce a new phenotype, which we call a \textit{replacer}.
Replacers are efficient competitors.
When a replacer is chosen for reproduction,  the offspring will always replace an individual of another type (if available).
We determine the selective advantage of replacers in well-mixed populations and on one-dimensional lattices.
We find that being a replacer substantially boosts the fixation probability of neutral and deleterious mutants.
In particular, fixation probability of a single neutral replacer who invades a well-mixed population of size $N$ is of the order of $1/\sqrt N$ rather than the standard $1/N$.
Even more importantly, replacers are much better protected against invasions once they have reached fixation.
Therefore, replacers dominate the mutation selection equilibrium even if the phenotype of being a replacer comes at a substantial cost:
curiously, for large population sizes and small mutation rates, the relative reproductive rate of a successful replacer can be as low as $1/e$.

\end{abstract}

\section*{Introduction}
The long-term fate of biological populations is shaped by evolutionary forces~\cite{ewens2004mathematical,barton2007evolution}.
Mutation creates new types.
Selection acts on those types.
Stochastic evolutionary dynamics is often described by variations of the Moran process~\cite{moran1958random,kimura1968evolutionary,Maruyama1974a}.
The population consists of $N$ individuals. 
Each type of individual has a certain reproductive rate.
The process proceeds in discrete time steps.
In each step, one individual is selected as a parent, with a probability proportional to its reproductive rate.
This parent then produces an offspring that migrates and replaces another individual in the population.

If we consider the Moran process for very low mutation rate we often ask: 
What is the probability that a single mutant with relative reproductive rate $r$ successfully invades and takes over a population consisting of $N-1$ residents with reproductive rate 1? \cite{kimura1962probability,maruyama1970fixation,whitlock2003fixation,patwa2008fixation}
This fixation probability depends on the values of $r$ and $N$.
For a neutral mutant ($r=1$) this fixation probability is $1/N$.
In general, the fixation probability is $(1-1/r)/(1-1/r^N)$.
Hence, for disadvantageous mutants ($r<1$), the fixation probability is exponentially small in $N$. For advantageous mutants ($r>1$), the fixation probability tends to the constant $1-\frac1r$ as the population size $N$ becomes large.
Thus, deleterious mutations are unlikely to fix in large populations, while beneficial mutations are reasonably likely to fix.
Note that fixation is never certain -- even beneficial mutants could be lost to random drift, especially if they are initially present in the population at low frequency~\cite{durrett1994importance,nowak2006evolutionary,nowak2010evolutionary}. 

The Moran process allows us to study competing phenotypes that differ in aspects other than a constant reproductive rate.
For example, selection could be frequency dependent;
then we study evolutionary game dynamics in finite populations~\cite{ohtsuki2006simple,nowak2006five,szabo2007evolutionary,santos2008social,broom2010evolutionary,szolnoki2014cyclic,allen2017evolutionary,perc2017statistical,su2019evolutionary,mcavoy2020social,broom2022game,su2023strategy,wang2024evolutionary}.
Or the population could reside on a certain structure; then we study evolutionary graph theory~\cite{lieberman2005evolutionary,broom2008analysis,diaz2014approximating,tkadlec2019population,tkadlec2021fast,diaz2021survey,kopfova2025colonization}.
Other possibilities are that the reproductive rates are location dependent~\cite{bulmer1972multiple,maciejewski2014environmental,svoboda2023coexistence} or that individuals differ in mutation rates or migration patterns~\cite{yeaman2011establishment,allen2012mutation,kaveh2019environmental,tkadlec2023evolutionary}.

Here we consider the evolutionary dynamics of a phenotype which we call a \textit{replacer}.
In the standard Moran process, the offspring of the reproducing parent might end up replacing an individual that has the same type as the parent.
Whenever this occurs, the reproductive step is effectively wasted, since the population composition has not changed (only time was gained).
We say that a phenotype is a replacer, if it avoids those situations whenever possible.
Thus, when a replacer is reproducing, the offspring replaces a random individual among those that have a different type than the parent.
It replaces the ``other type'' whenever such an individual is available (and within reach). 

While we are primarily interested in the evolution of replacers for the simplicity and elegance of the mathematical model, 
the replacer phenotype also occurs naturally.
We give examples of three scenarios.
First, during cancer evolution in Drosophila, the dMyc transcription factor transforms cells into super-competitors who actively trigger apoptosis in those adjacent that are not themselves super-competitors~\cite{moreno2004dmyc}.
Thus, the super-competitor cells behave like replacers.
Second, during fungal competition, hyphae may preferentially direct antagonistic growth toward neighboring genetically distinct hyphae, inhibiting, killing, or overgrowing them through hyphal interference and necrotrophic mycoparasitism~\cite{boddy2016fungal,jeffries1995biology}.
Third, in cultural evolution, e.g. in opinion dynamics, when attempting to spread their opinions, humans naturally concentrate their efforts on the contacts who have the other opinion (and avoid ``preaching to the choir'').

It is clear that being a replacer should be an advantageous trait.
But how much selective advantage is gained?
To answer this question, we calculate three measures of success for a mutant who is a replacer.
First, we calculate \textit{fixation probability}.
Being a replacer should increase the fixation probability.
Second, we calculate the \textit{elimination probability}, that is, the probability that a population consisting of $N-1$ individuals gets wiped out by an invader.
Being a replacer should decreases the elimination probability.
Third, we calculate the expected abundance in a mutation-selection equilibrium~\cite{antal2009mutation,tarnita2009mutation}.
Being a replacer should increase the expected abundance.

For well-mixed populations of size $N$, we find that the fixation probability of a neutral replacer is proportional to $1/\sqrt{N}$.
This quantity is substantially larger than the $1/N$ fixation probability of a neutral non-replacer.
For non-neutral mutants, the difference in the fixation probability is less pronounced:
For disadvantageous mutants ($r<1$), the fixation probability is still exponentially small in $N$. 
For advantageous mutants ($r>1$), the fixation probability tends to the same constant $1-\frac1r$ as for the non-replacers, when the population size $N$ becomes large.
However, we find that replacers get a major boost in terms of the elimination probability.
In particular, while disadvantageous non-replacers are eliminated with constant probability, for replacers the elimination probability is exponentially small. 
Thus, once replacers become established, they are much better protected against further invasions, even when they have a diminished reproductive rate.
As a consequence, we show that replacers dominate in the mutation-selection equilibrium even when they are disadvantageous.
Curiously, for large population size and small mutation rates, the relative reproductive rate of a successful replacer can be as low as $1/e\doteq 0.37$.
Moreover, those effects occur (and are often even more pronounced) for small population sizes $N$, or when the underlying spatial structure is a one-dimensional lattice rather than a well-mixed population.
In particular, on a large one-dimensional lattice the fixation probability of a neutral replacer approaches the value $1/3$, and the replacers dominate the mutation-selection equlibrium at low mutation rates even when their relative reproductive rate is as low as $0.5$.

\section*{Results}

\subsection*{Model}
Here we describe in detail the Moran process that governs the evolutionary dynamics, the phenotype of replacers, and the key quantities of fixation and elimination probability that characterize the evolutionary dynamics.

\paragraph{Evolutionary dynamics: Moran process.}
Throughout most of the paper, we consider a competition between mutants and residents in a well-mixed population of size $N$.
Residents have normalized reproductive rate 1, while mutants have relative reproductive rate $r$.
Mutations may be neutral ($r=1$), beneficial/advantageous ($r>1$), or deleterious/disadvantageous ($r<1$).
The evolutionary dynamics is governed by the \textit{Moran Birth-death process}.
The basic setting is as follows:
The evolution occurs in discrete time-steps.
In each step, one individual is selected for reproduction, where the probability of each individual being selected is proportional to their reproductive rate.
Then, another individual is chosen for death uniformly at random, and the dying individual is replaced by the copy of the reproducing individual.
When the reproducing and the dying individual are of the same type, nothing changes; otherwise the reproducing type becomes more frequent.
Such steps are repeated until all individuals in the population are of the same type (either all mutants or all residents).
This basic setting can be extended in several ways, see below.

\paragraph{Replacers.}
In this work, we consider mutants who are replacers. When a replacer is selected for reproduction, the dying individual is not chosen randomly from among all the other individuals, but only from among all the residents. Thus, a reproducing replacer never ``wastes its turn'', and will always cause the replacer frequency to increase.
For clarity, mutants who are not replacers are called \textit{oblivious} mutants.

\begin{figure}[h!]
    \centering
    \includegraphics[width=1\linewidth]{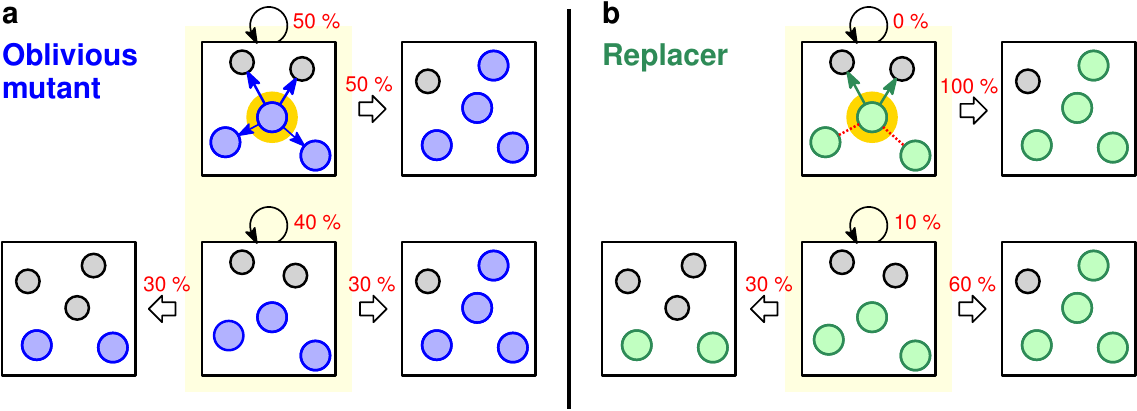}
    \caption{\textbf{Moran process with oblivious mutants and with replacers.} 
    In each step of the Moran Birth-death process, first a random individual is selected for reproduction, and then it displaces another individual. 
    \textbf{a,} In the standard Moran process, the mutants (blue) are oblivious, so a reproducing mutant always displaces another individual uniformly at random, including both mutants and residents (grey).
    Top row: When the middle mutant is reproducing (yellow circle), with probability 50~\% the mutant abundance increases and with probability 50~\% it stays the same.
    Bottom row: Taking into account all possible reproducing individuals gives the depicted percentages to increase, decrease, or maintain the mutant abundance.
    \textbf{b,} In contrast, the replacers (green) never ``waste their turn'', always replacing a random resident rather than a random individual.
    Top row: When the middle replacer is reproducing (yellow circle), the replacer abundance increases with probability 100~\%.
    Bottom row: Taking into account all possible reproducing individuals, the percentages to increase, decrease, or maintain the replacer abundance become as shown.
    }
    \label{fig:model}
\end{figure}

\paragraph{Key quantities: Fixation and elimination probability.}
To characterize the long-term fate of an evolving population we study the following two classic quantities. First, the \textit{fixation probability}, denoted $\fp_r(N)$, is the probability that Moran process starting with a single invading mutant with reproductive rate $r$ ends with all individuals being mutants. 
Fixation probability measures the ability of mutants to establish themselves.

Second, the \textit{elimination probability}, denoted $\ep_r(N)$, is the probability that Moran process starting with $N-1$ mutants ends with all individuals being residents.
Elimination probability quantifies the extent to which mutant populations are protected from subsequent resident invasions.

When the mutant is a replacer, we denote the fixation probability by $\fps_r(N)$ and the elimination probability by $\eps_r(N)$. When the mutant is neutral, we use a shorthand notation $\fp(N)$ to mean $\fp_{r=1}(N)$.

\paragraph{Extension: Mutation-selection process and lattice spatial structure.}
We consider two extensions of the standard Moran process.

First, we consider the mutation-selection process. 
In the mutation-selection process with a mutation rate $u\in(0,1)$, during each reproductive event there is a probability $u$ that a mutation occurs and the offspring of the reproducing individual becomes the other type than the parent. 
Therefore, no type ever gets established forever, and instead of studying fixation and elimination probabilities, we measure the average number (abundance) of mutants in the long term.

Second, Moran process can be adapted to accommodate spatial structures.
Here we consider a special case of a one-dimensional lattice with periodic boundary condition. That is, there are $N$ sites arranged along a cycle $C_N$, and any time an individual reproduces, the offspring migrates to one of the two neighboring sites.
If the reproducing individual is \textit{not} a replacer, the offspring migrates to a neighbor selected uniformly at random.
If the reproducing individual is a replacer, the offspring migrates to a random resident neighbor (or to a random neighbor, if no resident neighbor exists). The corresponding fixation and elimination probabilities are denoted $\fp_r(C_N)$, $\ep_r(C_N)$ for oblivious mutants, and $\fps_r(C_N)$, $\eps_r(C_N)$ for replacers.

\paragraph{Asymptotic notation.} To concisely express results about the fixation probability $\fp(N)$ and the elimination probability $\ep(N)$ for large population sizes $N$, we employ the following standard notation~\cite{cormen2022introduction}.
We write $A(N)\approx B(N)$ to denote that $A(N)$ and $B(N)$ are asymptotically equal, that is, $\lim_{N\to\infty} \frac{A(N)}{B(N)}=1$.
We write $A(N)\sim B(N)$ to denote that $A(N)$ is proportional to $B(N)$, that is, there exist positive constants $c_1$, $c_2$ such that $c_1\le \frac{A(N)}{B(N)}\le c_2$.
We write $A(N)\lesssim B(N)$ to denote that $A(N)$ is less than or proportional to $B(N)$, that is, there exists a positive constant $c_2$ such that $\frac{A(N)}{B(N)}\le c_2$.
So for example if $\fp(N)=\frac2{N-1}$ then we would write $\fp(N)\approx\frac2N$ and $\fp(N)\sim \frac1N$.

\subsection*{Neutral evolution}
In a well-mixed population of size $N$, the situation at any given time point of the evolutionary dynamics is completely described by just the number~$i$ of invading mutants.
To track how the number of mutants changes in time, for each $1\le i\le N-1$ we denote by $\p_i$ (resp.\ $\q_i$) the probability that the number of mutants increases (resp.\ decreases) in a single step of the Moran process.
Many aspects of the evolutionary process depend only on the ratios $\bias_i=\p_i/\q_i$ of those two probabilities. We call this ratio $\bias_i$ the \textit{forward bias}.

The forward bias $\bias_i$ in a well-mixed population can be expressed explicitly, for any population size $N$ and any number $i$ of mutants. %
When the mutants are neutral and both types of individuals are oblivious, %
standard calculation yields
\[ \p_i = \frac{i}{N}\cdot\frac{N-i}{N-1} \quad\text{and}\quad \q_i = \frac{N-i}{N}\cdot\frac{i}{N-1},
\]
thus $\bias_i = \frac{\q_i}{\p_i} = 1$, regardless of $N$ and $i$.
In other words, when both types are oblivious, the size of the mutant subpopulation is always as likely to increase as it is to decrease.

In contrast, when mutants are replacers, an analogous calculation yields
\[ \ps_i = \frac{i}{N}\cdot1 \quad\text{and}\quad \qs_i =\q_i= \frac{N-i}{N}\cdot\frac{i}{N-1},
\]
thus $\biass_i = \frac{\ps_i}{\qs_i} = \frac{N-1}{N-i}$.
Therefore, in this setting the forward bias of replacers increases with $i$.
In particular, for $i=1$ we have $\biass_1=1 =\bias_1$, meaning that being a replacer does not bring any advantage if there is only a single mutant.
At the other extreme, for $i=N-1$ we have $\biass_{N-1}=N-1$, so being a replacer helps immensely.
See \cref{fig:i} for a direct comparison.

\begin{figure}[h!]
    \centering
    \includegraphics[width=0.65\linewidth]{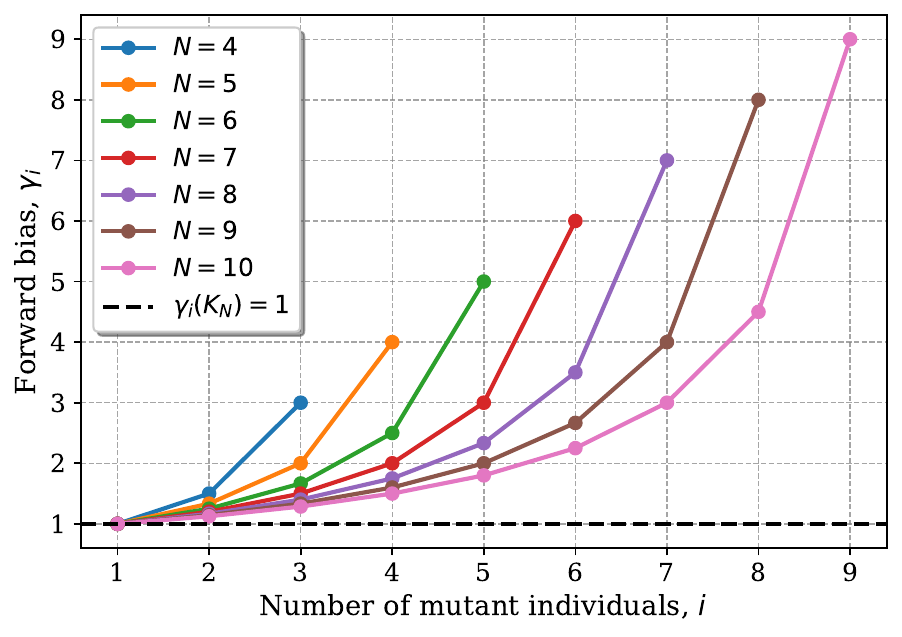}
    \caption{\textbf{Forward biases in well-mixed populations.}
    In the standard Moran process with $i$ oblivious mutants, the number of mutants is as likely to increase as it is to decrease, that is, the forward bias satisfies $\bias_i=1$ (black, dashed).
    When the mutants are replacers, the forward bias in a population of size $N$ satisfifes $\biass_i=\frac{N-1}{N-i}$, so it increases with the number $i$ of mutants, reaching up to $\biass_{N-1}=N-1$.
    Here population sizes are $N\in\{4,5,\dots,10\}$ (colors) and the number of mutants $i$ varies from 1 to $N-1$. %
    }
    \label{fig:i}
\end{figure}

The above differences have a profound effect on the fixation probability of a single mutant who attempts to invade a population of size $N$.
When the mutant is oblivious, its fixation probability is known to satisfy $\fp(N)=1/N$~\cite{nowak2006evolutionary}.
When the mutant is a replacer, we show that the fixation probability $\fps(N)$ is instead proportional to $1/\sqrt{N}$.

\begin{theorem}[Neutral fixation probability]\label{thm:fp-r1} In a well-mixed population of size $N$, the fixation probability $\fps$ of a single replacer satisfies $\fps(N)\approx \sqrt{\frac{2}{\pi\cdot N}}$.
\end{theorem}

Here the symbol $\approx$ means that for large $N$, the ratio of the two quantities tends to 1.
In particular, in the limit of large population size $N\to\infty$, the fixation probability $\fps(N)$ of a neutral replacer still decreases to 0, but it does so substantially more slowly than the fixation probability $\fp(N)$ of an oblivious mutant.

Another important quantity that depends only on the forward biases is the elimination probability of the mutants, once they become established. This measures the stability of the mutant population with respect to invasions of a single resident.

In the setting where both mutants and residents are oblivious, it is known that the elimination probability $\ep(N)$ 
of $N-1$ mutants, when facing a single resident, is again equal to $\ep=1/N$.
When the mutants are replacers, we show that the elimination probability $\eps(N)$ is exponentially small.

\begin{theorem}[Neutral elimination probability]\label{thm:ep-r1} In a well-mixed population of size $N$, the elimination probability $\eps(N)$ of $N-1$ replacers satisfies $\eps(N)\approx 2/e^{N-1}$.
\end{theorem}

Thus, once established, the replacer subpopulation is extremely unlikely to be wiped out by the oblivious residents.
The results are summarized in~\cref{fig:r1}a (up to multiplicative constants) and plotted in~\cref{fig:r1}b.

\begin{figure}[h!]
    \centering
    \includegraphics[width=0.35\linewidth]{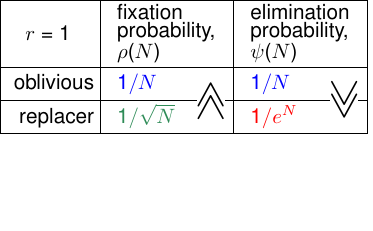}
    \includegraphics[width=0.55\linewidth]{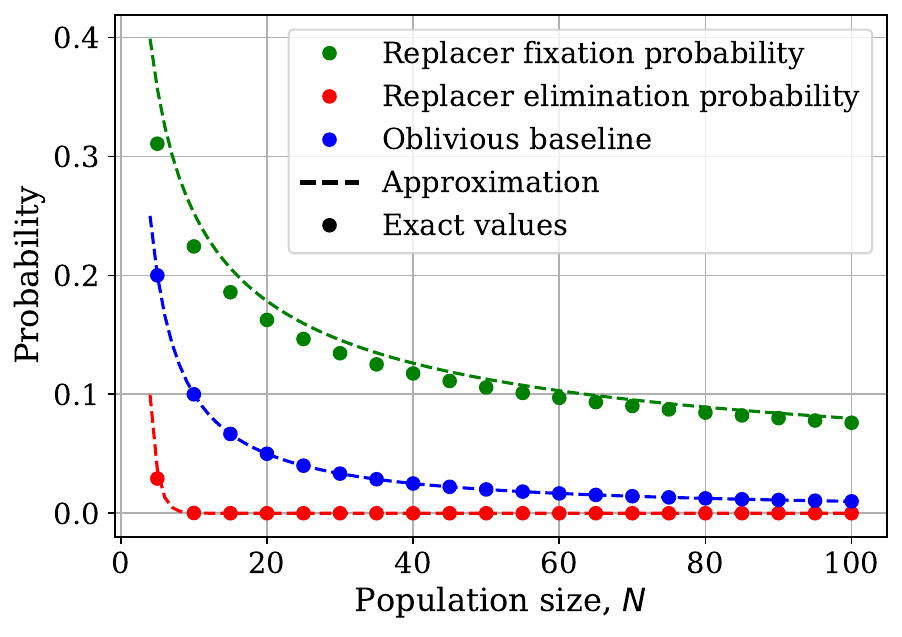}
    \caption{\textbf{Fixation and elimination probabilities of neutral replacers.}
    \textbf{a,} Neutral replacers have a fixation probability of the order of $\fps(N)\sim 1/\sqrt N$, which is substantially higher than the fixation probability $\fp(N)=1/N$ of standard (oblivious) mutants, depicted by a $\gg$ sign in the table.
    Moreover, replacers have an  exponentially small elimination probability of the order of $\eps(N)\sim 1/e^N$, which is substantially smaller than the elimination probability $\ep(N)=1/N$ of standard (oblivious) mutants, depicted by a $\ll$ sign in the table.
    \textbf{b,} While all fixation and elimination probabilities tend to 0 as the population size $N$ grows large, they do so at very different rates. For oblivious mutants both the fixation and the elimination probability are equal to $\fp(N)=\ep(N)=1/N$, which serves as a natural baseline (blue). The fixation probability $\fps(N)$ of replacers is substantially larger (green), and their elimination probability $\eps(N)$ is substantially smaller (red).
    The asymptotic formulas (dashed) closely match the exact values (dots), even for small population sizes. %
    }
    \label{fig:r1}
\end{figure}

\subsection*{Non-neutral evolution}
Not all mutations are neutral -- some increase the individual's reproductive rate while others decrease it.
The Moran process accommodates for such effects by assigning relative reproductive rate $r\ne 1$ to each mutant, as compared to a baseline reproductive rate $1$ assigned to each resident.
The case $r>1$ corresponds to advantageous (beneficial) mutations, the case $r<1$ corresponds to disadvantageous (deleterious) mutations.

When both the mutants and the residents are oblivious, standard calculation yields that the forward bias satisfies $\bias_i = \frac{\p_i}{\q_i} = r$, regardless of the population size $N$ and the number $i$ of mutants.
As a consequence, the fixation probability of a single mutant with relative reproductive rate $r$ is $\fp_r(N)=\frac{1-\frac1r}{1-\frac1{r^N}}$.
In particular, when $r>1$ and the population size $N$ is large, the fixation probability tends to a positive constant $1-1/r$.
In contrast, when $r<1$, the fixation probability is roughly $r^N$, that is, it tends to 0 exponentially quickly.
Thus, in large populations, beneficial mutations have a constant chance of fixing, whereas deleterious mutations are exponentially unlikely to fix.

Somewhat surprisingly, both those claims remain true also when the invading mutant with reproductive rate $r$ is not oblivious, but rather a replacer. We can prove the following result.

\begin{theorem}[Fixation probability for large $N$]\label{thm:fp-r}
In a large well-mixed population of size $N$, the fixation probability $\fps(N)$ of a single replacer with relative reproductive rate $r$ satisfies
\begin{align*}
\fps_r(N) &\approx 1-\frac1r \text{\ \  if $r>1$ and }\\
\fps_r(N) &\sim 
\frac{(c_r)^N}{\sqrt{N}} \text{\ \  if $r<1$},
\end{align*}
where $c_r=r\cdot e^{1-r}>r$ is a constant that depends on $r$.
\end{theorem}
Note that when $r<1$, both $\fp_r(N)$ and $\fps_r(N)$ are exponentially small in $N$.
However, the ratio $\fps_r(N)/ \fp_r(N)$ grows unbounded since $c_r/r>1$.
Thus, the fixation probability $\fp_r(N)$ of a deleterious oblivious mutant decays towards 0 exponentially faster than the fixation probability $\fps_r(N)$ of an equally deleterious replacer.

As in the case of neutral mutations, we also consider the elimination probability of $N-1$ mutants with reproductive rate $r\ne 1$.
For oblivious mutants, the elimination probability can be computed in the same fashion as the fixation probability.
In particular, when $r>1$ and the population size $N$ is large, the elimination probability of $N-1$ mutants with reproductive rate $r$ is roughly $1/r^N$, that is, it tends to 0 exponentially quickly.
In contrast, when $r<1$, the elimination probability of $N-1$ mutants with reproductive rate $r$ tends to a positive constant $1-r$.
Thus, in large populations, established beneficial mutations are strongly protected, whereas established deleterious mutations are highly susceptible to resident invasions.

In contrast, we show that mutants who are replacers, are strongly protected regardless of whether they are beneficial or deleterious.

\begin{theorem}[Elimination probability for large $N$]\label{thm:ep-r}
In a large well-mixed population of size $N\ge 3$, the elimination probability $\fps$ of a $N-1$ replacers with relative reproductive rate $r$ satisfies
\begin{align*}
\eps_r(N) &\lesssim 1/r^N\text{\ \  if $r>1$ and }\\
\eps_r(N) &\sim (c'_r)^N  \text{\ \  if $r<1$},
\end{align*}
where $c'_r=1/e^r<1$ is a constant that depends on $r$.
\end{theorem}

The results of \cref{thm:fp-r,thm:ep-r} are summarized in \cref{fig:table}.

\begin{figure}[h!]
    \centering
    \includegraphics[width=0.8\linewidth]{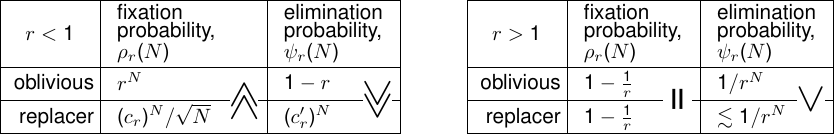}
    \caption{\textbf{Fixation and elimination probabilities of non-neutral replacers.}
    \textbf{a,} When the mutation is deleterious, the fixation probability is exponentially small regardless of whether the mutant is oblivious or a replacer. However, for replacers the base of the exponential is larger, so fixation is exponentially more likely. Perhaps more importantly, the elimination probability of replacers is also exponentially small, meaning that once established, the replacers are well protected against resident invasions. This strongly contrasts with a constant elimination probability of established oblivious mutants.
    \textbf{b,} When the mutation is beneficial, being a replacer does not affect the fixation probability when the population size $N$ is large. However the replacers again do have diminished elimination probability, as compared to oblivious mutants.
    }
    \label{fig:table}
\end{figure}

\subsection*{Mutation-selection equilibrium}
Understanding both the fixation and the elimination probability of the replacers allows us to analyze the mutation-selection process, where each time an individual reproduces, with probability $u\in(0,1)$ the offspring mutates to the opposite type than the parent.
For a given mutant reproductive rate $r$ and a given mutation rate $u$, we study the expected number $\lambda$ of mutants in the population, averaged over a long time frame. We find that replacers fare much better than the oblivious mutants with the same reproductive rate, see~\cref{fig:balance}.

\begin{figure}[h!]
    \centering
    \includegraphics[width=0.45\linewidth]{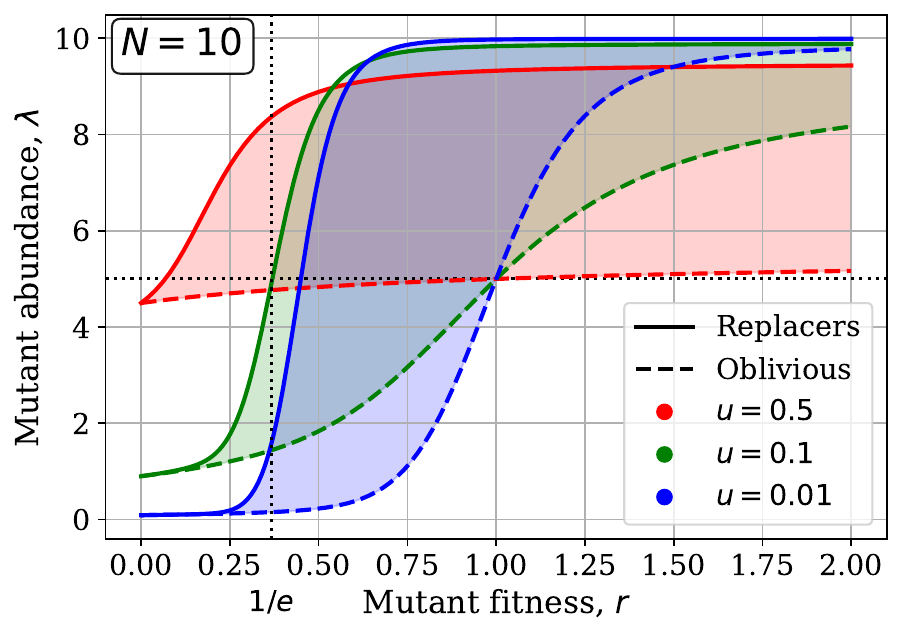}
    \includegraphics[width=0.45\linewidth]{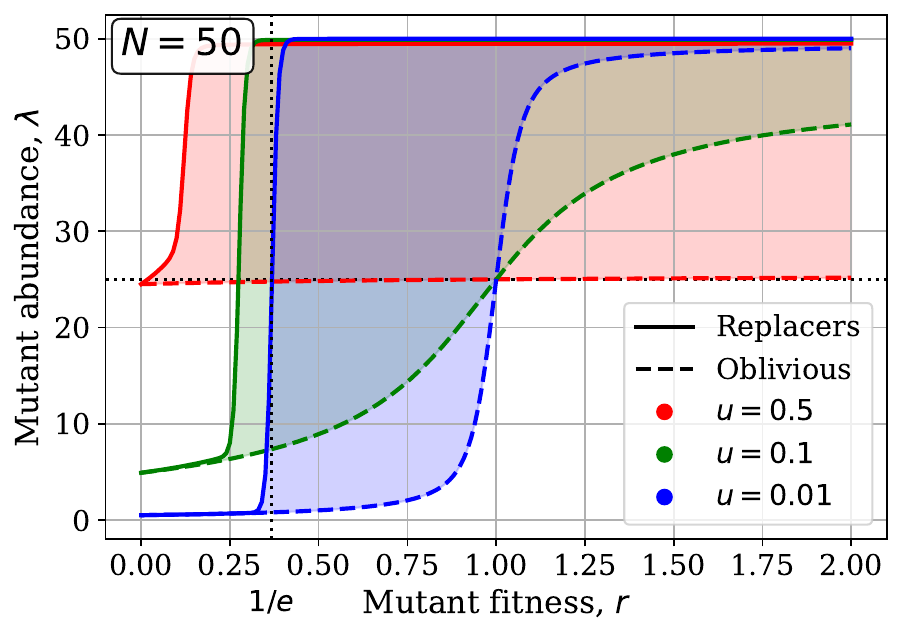}
    \caption{\textbf{Proportion of mutants in the mutation-selection process.}
    In the mutation-selection process, any time an individual reproduces, the offspring mutates to the other type with fixed probability $u\in(0,1)$.
    The long-term average mutant abundance $\lambda$ depends on the mutant reproductive rate $r$ ($x$-axis) and also on the mutation rate $u$ (colors).
    When the mutants are oblivious (dashed lines), they form a majority in the population when $r>1$, otherwise they are in a minority. 
    In contrast, replacers (solid lines) form majority even for certain values $r<1$, e.g., for $r\ge 0.46$ when $N=10$ (blue solid line vs dotted horizontal line in the left panel). The differences in abundance are indicated by the shaded areas.
    In particular, we show that when $u\to 0$ and $N$ is large, replacers dominate so long as their relative reproductive rate is above a threshold $r^\star=1/e\doteq 0.37$ (dotted line).
    For other $u$ they dominate for even smaller values of $r$.
    Here the population size is \textbf{a,} $N=10$, \textbf{b,} $N=50$.
    }
    \label{fig:balance}
\end{figure}

Note that in the limit of low mutation rate $u\to 0$, the average composition of the whole population is given by the ratio of the fixation probability and the elimination probability. When $r<1$, both those probabilities are exponentially small in the population size $N$, but the bases of the exponentials differ.
For each population size $N$, we identify the threshold value $r^\star$ that determines whether mutants or residents are more prevalent at the mutation-selection equilibrium.
For $r>r^{\star}$, the advantage obtained by being a replacer outweighs the decrease in the mutant reproductive rate and the population typically consists almost exclusively of mutants.
On the other hand, for $r<r^\star$ the decrease in the mutant reproductive rate is so large that even being a replacer can not compensate and the population consists almost exclusively of residents.

In particular, we show that for $N\to \infty$, the threshold value $r^\star$ tends to $1/e\doteq 0.37$, see \SI{} for details.
Moreover, the limit behavior kicks in early:
We prove that for $N>6$ the threshold value satisfies $r^\star<0.5$, and the right panel of \cref{fig:balance} shows that for $N=50$ the solid blue line crosses the horizontal dotted line almost exactly at $r^\star=1/e$.

\subsection*{The replacer cost}
It is conceivable that in practice the ability to express the replacer phenotype comes with an associated cost that can be represented as a reduction in the mutant reproductive rate.
Here we study how much reproductive rate is a mutant willing to sacrifice, in order to become a replacer.
More formally, what is the cost $s$ such that a replacer with relative reproductive rate $r=1-s$ performs equally well as an oblivious mutant with reproductive rate $r=1$?
In reality, this cost $s$ depends on the population size $N$, so we denote it by $s_N$.
And it also depends on whether we measure the mutant performance in terms of the fixation probability, or the elimination probability, or the abundance in the mutation-selection process.

In \cref{fig:costs} we consider those three cases.
For the fixation probability, we show that as the population size $N$ grows large, the cost $s_N$ tends to zero. 
However, we show that for a broad range of population sizes $N$ this cost is quite substantial.
In particular, the replacer cost peaks for $N=8$, %
where it satisfies $s_8>0.27$. For $N=100$ it is still substantial, satisfying $s_{100}>0.16$, see~\cref{fig:costs}a.

\begin{figure}[h!]
    \centering
\includegraphics[width=0.3\linewidth]{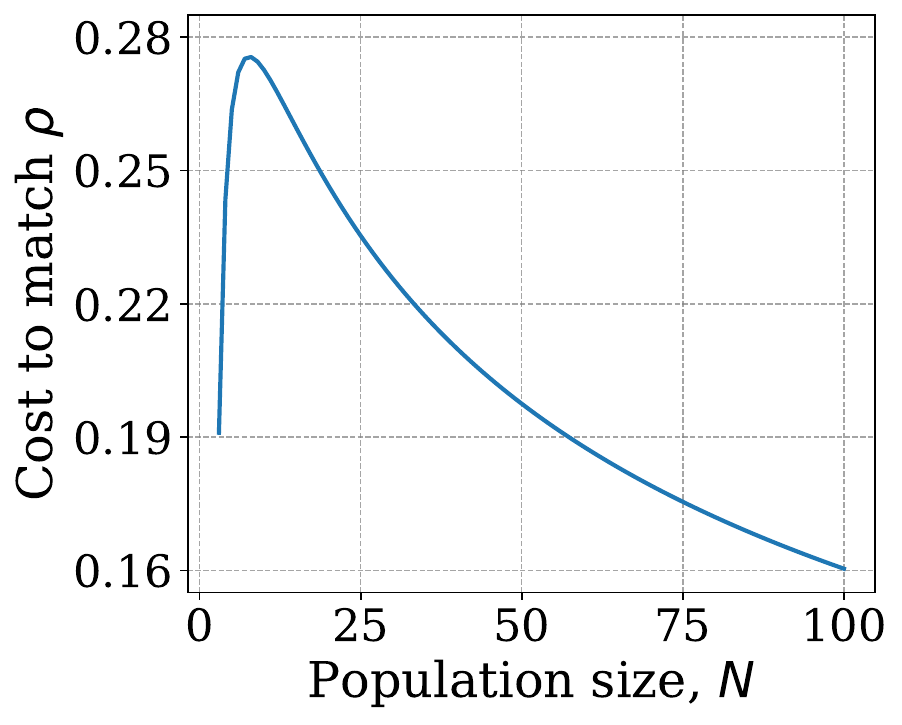}
\includegraphics[width=0.3\linewidth]{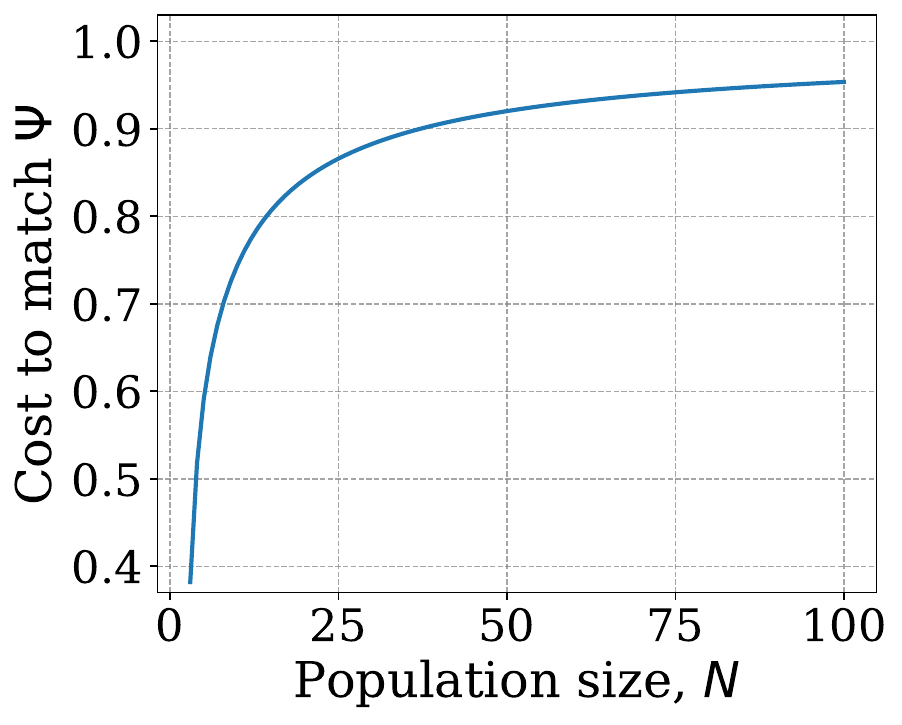}
\includegraphics[width=0.3\linewidth]{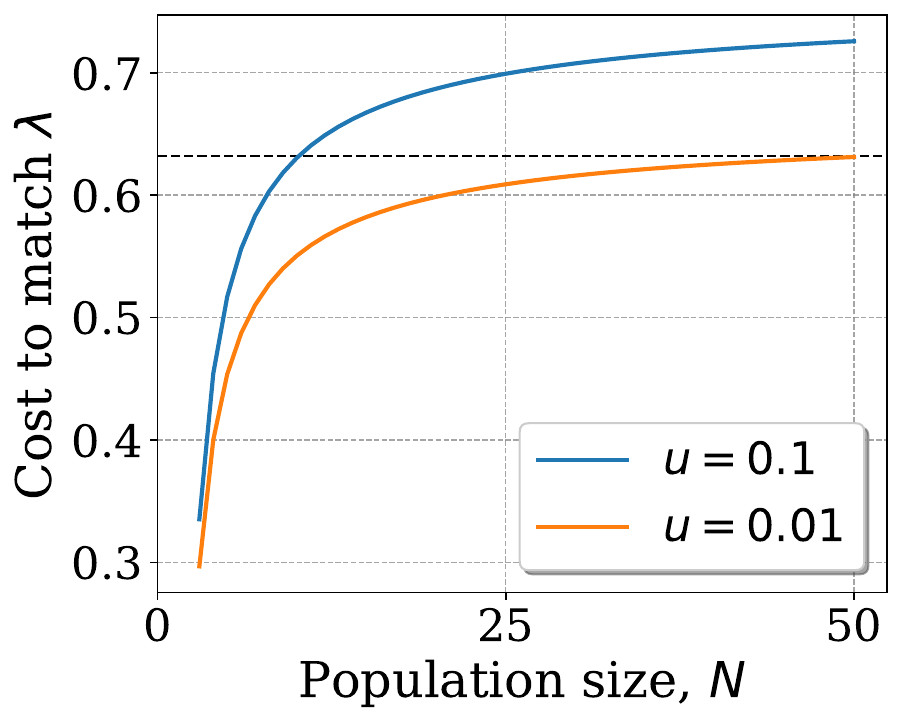}
    \caption{\textbf{The replacer cost.} What is the cost $s_N$ such that a replacer with reproductive rate $1-s_N$ achieves the same performance as a neutral oblivious mutant, in a population of size $N$? The answer depends on the measure of success.
    \textbf{a,} To match the fixation probability, a replacer is willing to decrease its reproductive rate by as much as $s_8>0.27$ when $N=8$, and by as much as $s_{100}>0.16$ when $N=100$. For large population sizes $N$, we have $s_N\to 0$.
    \textbf{b,} To match the elimination probability, a replacer is willing to sacrifice almost all its reproductive rate, and we have $s_N\to1$.
    \textbf{c,} To match the abundance $\lambda$ at the mutation-selection equilibrium, we have $s_N\to 1-1/e\doteq 0.63$ when the mutation rate $u$ is small. 
    }
    \label{fig:costs}
\end{figure}

When it comes to elimination probability, the replacer cost gradually increases as the population size $N$ grows large, and approaches a theoretical limit $s_N\to 1$.
In other words, even extremely disadvantageous replacers manage to match the elimination probability of neutral non-replacers.

Similar results hold for the mutant abundance at the mutation-selection equilibrium. Increasing the population size enables replacers to remain competitive even for larger costs. In the limit of low mutation rate $u\to 0$, the highest possible cost is $s_N\to 1-1/e\doteq 0.63$.

\subsection*{Cycle spatial structure}
All our results so far have been for a well-mixed population.
In this section, we consider a spatially structured population where the $N$ sites are arranged along a cycle $C_N$, and the offspring always migrates to one of the two adjacent sites. Thus, the spatial structure is a one-dimensional lattice with a periodic boundary condition. We denote the corresponding fixation and elimination probabilities by $\fp_r(C_N)$, $\fps_r(C_N)$, $\ep_r(C_N)$, and $\eps_r(C_N)$.

For oblivious mutants, living on a cycle $C_N$ or in a well-mixed population does not make a difference when it comes to fixation and elimination probabilities.
That is, the top rows of \cref{fig:r1,fig:table} apply to the lattice $C_N$ as well as to the well-mixed population~\cite{nowak2006evolutionary}.
In particular, the fixation and elimination probability are equal to $1/N$ when $r=1$, otherwise they either tend to a constant, or they are exponentially small as the population size $N$ becomes large.

If the mutant is a replacer, the spatial structure does make a difference.
We are able to compute the fixation and elimination probabilities exactly, see~\SI{} for details.
It turns out that the threshold value of the reproductive rate shifts to $r^\star=1/2$.
When $r^\star=1/2$, both the fixation probability $\fps_{r=1/2}(C_N)$ and the elimination probability $\eps_{r=1/2}(C_N)$ are of the order of $1/N$.
The following result lists the approximate probabilities also in the cases $r<1/2$ and $r>1/2$.

\begin{theorem}[Cycle]
    Let $C_N$ be the cycle of size $N$, where $N$ is large. Then
\begin{align*}
\fps_r(C_N) &= \frac1{2N-1}
&\textrm{and }&&
\eps_r(C_N) &= \frac1{N-\frac12}
\text{\ \  if $r=1/2$, }\\
\fps_r(C_N) &\approx 1-\frac1{r+\frac12}
&\textrm{and }&& 
\eps_r(C_N) &\sim 1/(2r)^N  
\text{\ \  if $r>1/2$, }\\
\fps_r(C_N) &\sim (2r)^N  
&\textrm{and }&&
\eps_r(C_N) &\approx 1-2r
\text{\ \  if $r<1/2$}.
\end{align*}
\end{theorem}

In particular, for large population sizes $N\to\infty$, being a replacer on a cycle $C_N$ has a comparable effect to receiving a certain boost in terms of the reproductive rate.
When the mutant has reproductive rate $r<1/2$, an oblivious mutant would have fixation probability $\fp_{r<1/2}(C_N)\sim r^N$ (cf. \cref{fig:table}), whereas a replacer has fixation probability $\fps_{r<1/2}(C_N)\sim (2r)^N$. In other words, becoming a replacer effectively doubles the mutant's reproductive rate ($r\to 2r$).
The same boost $r\to 2r$ occurs for the elimination probability, regardless of the reproductive rate $r$.
However, when the mutant has a reproductive rate $r>1/2$ the boost for the fixation probability is different.
In that case, a replacer has fixation probability $\fps_{r}(C_N)\approx 1-\frac{1}{r+\frac12}\approx\fp_{r+\frac12}(C_N)$.
Therefore, becoming a replacer contributes an additive $+\frac12$ to the mutant's reproductive rate ($r\to r+\frac12$). 
Note that for $r\in(\frac12,1)$ and large population size $N$ an oblivious mutant has exponentially small fixation probability, whereas the fixation probability of a replacer tends to a positive constant.
See \cref{fig:boost} for an illustration.
In particular, the fixation probability $\fps_{r=1}(C_N)$ of a single neutral replacer who is invading a large cyclically structured population of size $N$, tends to $\fps_{r=1}(C_N)=1/3$.

\begin{figure}[h!]
    \centering

    \includegraphics[width=0.65\linewidth]{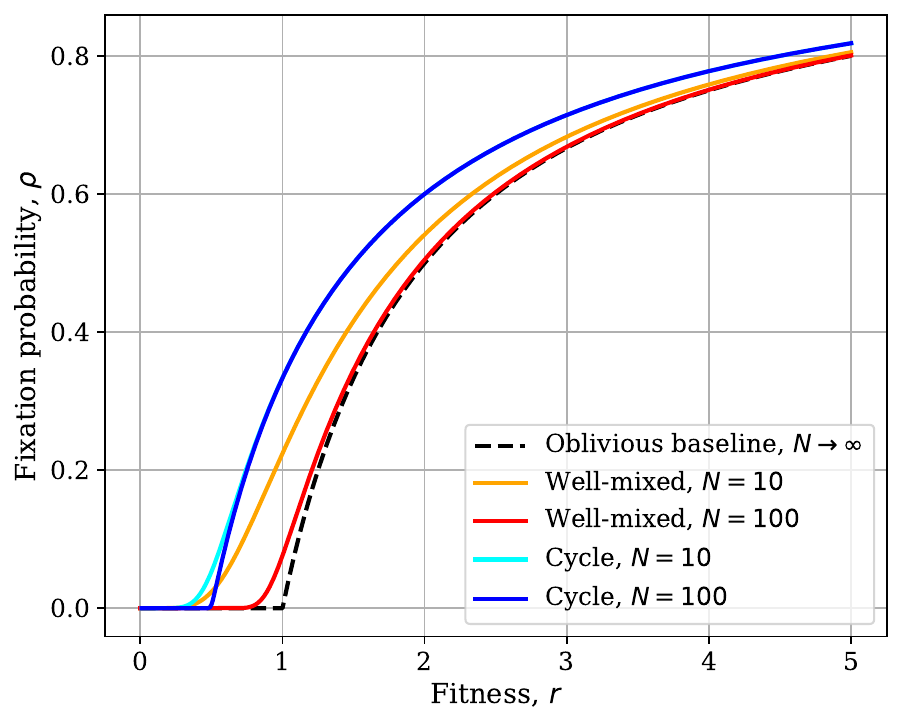}
    \caption{\textbf{Replacers on the cycle $C_N$.}
    For a single replacer, the fixation probability $\fps_r(N)$ on a well-mixed population (red, orange), and also the fixation probability $\fps(C_N)$ on a cycle $C_N$ (cyan, blue) are both increased, as compared to the baseline $\fp(N)$ given by a single oblivious mutant (black, dashed) in either of the two settings.
    For well-mixed populations, the fixation probability $\fp(N)$ tends to the natural baseline $\fp(N)$ as $N\to\infty$, for every mutant reproductive rate $r$ (see \cref{thm:fp-r1,thm:fp-r}).
    For Cycles $C_N$, the fixation probability $\fps(C_N)$ tends to be baseline only when $r<1/2$.
    For $r>1/2$, being a replacer effectively increases the mutant reproductive rate by $\frac12$, as seen by a horizontal offset of $0.5$ between the solid blue curve and the black dashed curve.
    Here $N\in\{10,100\}$ and $0\le r\le 5$.
    }
    \label{fig:boost}
\end{figure}

The above results imply that replacers dominate the mutation-selection equilibrium if they have reproductive rate $r>1/2$, and they are dominated if they have reproductive rate $r<1/2$, see \cref{fig:balance-cn}.

\begin{figure}[h!]
    \centering
    \includegraphics[width=0.8\linewidth]{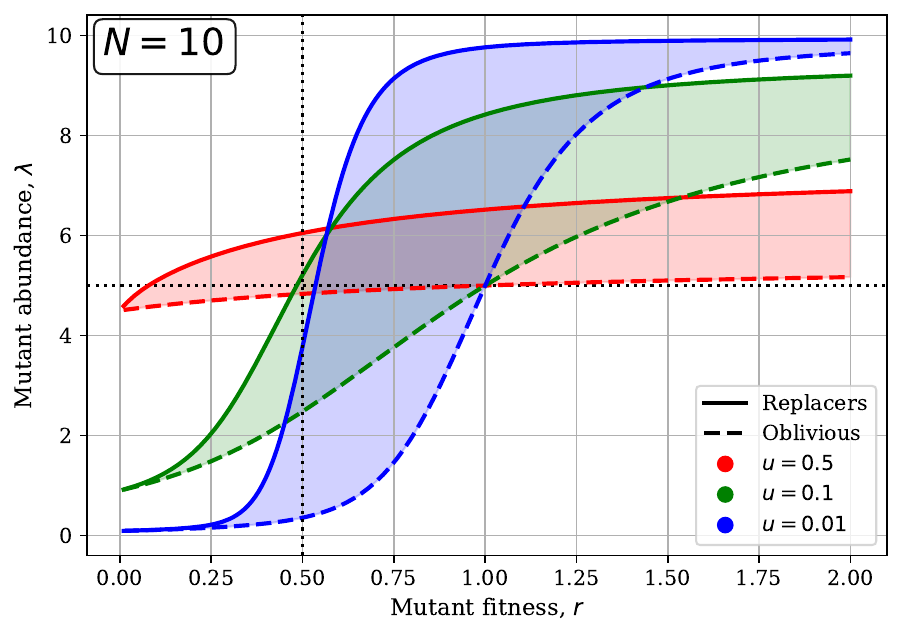}
    \caption{\textbf{Mutation-selection process on the cycle structure $C_N$.}
    The long-term average mutant abundance $\lambda$ in the cycle population $C_N$ depends on the mutant reproductive rate $r$ ($x$-axis) and also on the mutation rate $u$ (colors).
    When the mutants are oblivious (dashed lines), they form a majority in the population whenever $r>1$, otherwise they are in a minority. 
    In contrast, replacers (solid lines) form majority even for certain values $r<1$. The differences in abundance are indicated by the shaded areas.
    In particular, we show that for small mutation rates $u\to0$ and large population sizes $N$, replacers dominate so long as their relative reproductive rate is above a threshold $r^\star=1/2$ (dotted line).
    }
    \label{fig:balance-cn}
\end{figure}

\section*{Discussion}
In summary, we have studied stochastic evolutionary dynamics of replacer phenotypes
in well-mixed populations, which can be represented by complete graphs, and on once-dimensional lattices that can be represented by cycles.
The offspring of a standard, oblivious individual replaces a random neighbor in the population.
If the replaced individual is of the same type, then the reproductive event was wasted in so far as it did not change the composition of the population.
In contrast, the offspring of a replacer phenotype replaces a random neighbor of the other type (if available).
Therefore, replacers are more efficient.
They only waste reproductive events if the other type is not available in the neighborhood of the reproducing individual.
On a complete graph this never happens unless the population is homogeneous.

We find that replacers are efficient competitors.
The fixation probability of a single neutral replacer in a population of size $N$ scales as roughly $1/\sqrt N$.
In comparison, the fixation probability of a standard, oblivious neutral mutant is $1/N$.
Moreover, the probability that the replacers get eliminated once they become established in the population is exponentially small.
Therefore, replacers enjoy strong evolutionary stability.
Because of this stability, replacers tend to dominate oblivious types in the mutation-selection equilibrium, even if their relative reproductive rate is much lower.
When the evolutionary dynamics takes place on a cyclic structure, then replacers remain prevalent in the mutation-selection equilibrium for relative reproductive rate $r$ as low as $r\approx 1/2$.

In this paper, we have concentrated on two biologically relevant population structures: the well-mixed population and the one dimensional lattice (the cycle). 
Studying replacers on other spatial structures is an interesting direction of future research.
We note that in a well-mixed population and on a cycle it suffices to track the number of mutants.
On general spatial structures it is crucial to track the exact configuration of the population (the coloring of the graph).
Therefore, it becomes a formidable problem.

Curiously, it turns out that for both the well-mixed population and for the cycle, the evolutionary dynamics between neutral replacers and oblivious individuals closely resembles the dynamics of an evolutionary game with a payoff matrix.
For the well-mixed population, consider the game with matrix %
$\bordermatrix{ & R & O \cr
      R & 1 & 1 \cr
      O & 0 & 1 }.
$
That is, replacer receives payoff 1 from replacer and from oblivious individuals, while the oblivious individual receives payoff 1 only from the oblivious individuals.
The payoff of each individual determines its reproductive rate or fitness. %
Individuals are selected for reproduction with a probability proportional to their fitness.
It can be shown that the fixation probabilities of the respective types in this evolutionary game are equal to the fixation probabilities of replacers and oblivious individuals as described above (see~\SI{} for details).
Similarly, the dynamics of replacers and oblivious individuals on a cycle can be interpreted as a game with matrix
$\bordermatrix{ & R & O \cr
      R & 3 & 1 \cr
      O & 1 & 1 }.
$
But for general graphs, the evolutionary dynamics of replacers and oblivious types cannot be described by a game with a simple $2\times 2$ payoff matrix.

Finally, we mention one more direction for future work.
So far we have assumed that the offspring of a replacer migrates to a random site among the neighbors that are occupied by the other type.
But with even higher levels of sensing, the offspring could migrate to specific sites, perhaps the ones that are ``most beneficial'' for maximizing either the fixation probability or evolutionary stability of these ``smart replacers''.

\section*{Data and code availability}
The data and code are available at \url{https://github.com/mpecho19/The-evolutionary-advantage-of-replacers-in-the-Moran-process}.

\section*{Acknowledgment}
J.T. and M.P. were supported by Charles Univ.\ projects UNCE 24/SCI/008 and PRIMUS 24/SCI/012.
J.T. was supported by GA\v{C}R grant 25-17377S.


\newpage
\setcounter{theorem}{0}
\section*{Appendix}\label{sec:organization}

This is the Appendix for the paper \textit{\maintitle}. We provide the rigorous mathematical proofs of the theorems from the main text, namely:

\begin{theorem}[Neutral fixation probability]\label{thm:1} In a well-mixed population of size $N$, the fixation probability $\fps$ of a single replacer satisfies $\fps(N)\approx \sqrt{\frac{2}{\pi\cdot N}}$.
\end{theorem}

\begin{theorem}[Neutral elimination probability]\label{thm:2} In a well-mixed population of size $N$, the elimination probability $\eps(N)$ of $N-1$ replacers satisfies $\eps(N)\approx 2/e^{N-1}$.
\end{theorem}

\begin{theorem} [Fixation probability for large $N$]\label{thm:3-fp-r}
In a large well-mixed population of size $N$, the fixation probability $\fps(N)$ of a single replacer with relative fitness $r$ satisfies
\begin{align*}
\fps_r(N) &\approx 1-\frac1r \text{\ \  if $r>1$ and }\\
\fps_r(N) &\sim 
\frac{(c_r)^N}{\sqrt{N}} \text{\ \  if $r<1$},
\end{align*}
where $c_r=r\cdot e^{1-r}>r$ is a constant that depends on $r$.
\end{theorem}

\begin{theorem}[Elimination probability for large $N$]\label{thm:4-ep-r}
In a large well-mixed population of size $N\ge 3$, the elimination probability $\fps$ of a $N-1$ replacers with relative fitness $r$ satisfies
\begin{align*}
\eps_r(N) &\lesssim 1/r^N\text{\ \  if $r>1$ and }\\
\eps_r(N) &\sim (c'_r)^N  \text{\ \  if $r<1$},
\end{align*}
where $c'_r=1/e^r<1$ is a constant that depends on $r$.
\end{theorem}

\begin{theorem}
    [Cycle]\label{thm:5-cycle}
    Let $C_N$ be the cycle of size $N$, where $N$ is large. Then
\begin{align*}
\fps_r(C_N) &= \frac1{2N-1}
&\textrm{and }&&
\eps_r(C_N) &= \frac1{N-\frac12}
\text{\ \  if $r=1/2$, }\\
\fps_r(C_N) &\approx 1-\frac1{r+\frac12}
&\textrm{and }&& 
\eps_r(C_N) &\sim 1/(2r)^N  
\text{\ \  if $r>1/2$, }\\
\fps_r(C_N) &\sim (2r)^N  
&\textrm{and }&&
\eps_r(C_N) &\approx 1-2r
\text{\ \  if $r<1/2$}.
\end{align*}
\end{theorem}

The rest of this text is organized as follows.
In Section~\ref{sec:prelims} we define some important notions, and recall some standard results that we use later in our proofs.
In Section~\ref{sec:neutral} we prove \cref{thm:1} and \cref{thm:2}. 

In Section~\ref{sec:non-neutral} we prove \cref{thm:complete_fix_disadv,thm:complete_fix_adv}, that together immediately imply \cref{thm:3-fp-r}.   And we also prove \cref{thm:elim_complete_r_small,thm:elim_complete_r_big} that together immediately imply \cref{thm:4-ep-r}.

Finally, in Section~\ref{sec:cycle-structure} we prove \cref{thm:cycle_fixation,thm:cycle_elimination}.
Those two results, together with other arguments, imply \cref{thm:5-cycle}.

\section{Preliminaries}\label{sec:prelims}
We start with a few definitions. 

We consider a population spread over $N$ sites, where each site is occupied by a single individual.

A spatial structure of the population can be captured by a graph (network) $G$ that prescribes which pairs of sites neighbor each other.
We consider two possible spatial structures.
A well-mixed population of size $N$ is represented by a complete graph $K_N$, where every two sites neighbor each other.
The cycle spatial structure is represented by a cycle graph $C_N$, where each site has two neighbors (one clockwise, one counter-clockwise).

Initially, some individuals are mutants with relative reproductive rate (fitness) $r>0$, while the other individuals are residents with reproductive rate (fitness) 1.

Moran process changes the initial composition of the population as follows.
In each discrete time-step, one individual from the population is selected as a parent (with probability proportional to its fitness), the parent produces an offspring, and the offspring migrates to one of the neighboring sites.

\begin{defn}[Oblivious mutants and replacers]
We call an individual \emph{oblivious} if its offspring migrates to a neighboring site selected uniformly at random from among all the neighboring sites.
We call an individual a \emph{replacer} if the offspring migrates to a random neighboring site from among those sites that are currently occupied by the other type than the reproducing parent (if any such site is available).
\end{defn}

In this paper, we consider mutants who are either oblivious or replacers. The residents are always oblivious. 

In the absence of mutation, Moran process terminates with the mutants either fixating (that is, spreading over all the sites of the graph $G$), or being wiped out.

\begin{defn}[Fixation and elimination probability]
We define the (mutant) \emph{fixation probability} as the probability that starting with a single mutant and $N-1$ residents, the mutants fixate.
We define the (mutant) \emph{elimination probability} as the probability that starting with $N-1$ mutants and a single resident, the mutants get wiped out.
 
\end{defn}

To concisely talk about the fixation and elimination probability of oblivious mutants or replacers, in either the well-mixed population $K_N$ or the cycle spatial structure $C_N$, we employ the following notation.
Note that we consider only two types of graphs $G_N$, namely the complete graphs $K_N$ and the cycle graphs $C_N$. 

\begin{defn}
Let $G_N$ be a graph on $N$ vertices and let $r>0$ be the mutant fitness. Then:
\begin{itemize}

\item $\fp_r(G_N)$ denotes the fixation probability of a single \emph{oblivious} mutant in graph $G_N$ against $N-1$ \emph{oblivious} residents.

\item $\fps_r(G_N)$ denotes the fixation probability of a single \emph{replacer} mutant in graph $G_N$ against $N-1$ \emph{oblivious} residents.

\item $\ep_r(G_N)$ denotes the elimination probability of mutants when a single \emph{oblivious} resident starts in graph $G_N$ against $N-1$  \emph{oblivious} mutants.

\item $\eps_r(G_N)$ denotes the elimination probability of mutants when a single \emph{oblivious} resident starts in graph $G_N$ against $N-1$ \emph{replacer} mutants.

\end{itemize}
\end{defn}

In both spatial structures $K_N$ and $C_N$, the current composition of the population can be fully described by the number $i$ of mutants.
Thus, we can represent the Moran process as a Markov chain on a set of states labeled from 0 up to $N$, where in each step the number $i$ of mutants may change to $i-1$ or $i$ or $i+1$.
The main text uses a notion of a forward bias. Here we define its reciprocal, the backward bias.

\begin{defn}[Backward bias]\label{def:gamma}
Let $p^+_i$ be the probability of moving from state $i$ to state $i+1$ and 
$p^-_i$ the probability of moving from state $i$ to state $i-1$. We define
\[
\gamma_i = \frac{p^-_i}{p^+_i}.
\]
\end{defn}

\begin{defn}
    Let $T_k$ be the total fitness of states with $k$ mutants. Formally, $T_k = k\cdot r + (N-k)$.
\end{defn}

We will also make use of the following well known result.

\begin{lemma}[One-dimensional random walk]\label{lem:one-dimensional} 
Consider a one-dimensional random walk with states $\{0,1,\dots,N\}$ and absorbing boundaries at $0$ and $N$. 
Given $\gamma_1,\dots,\gamma_{N-1}$ as in Definition~\ref{def:gamma}, the probability of absorption at $N$ 
starting from state $i$ is
\[
\fp = \frac{1+\sum_{j=1}^{i-1}\prod_{k=1}^{j}\gamma_k}{1 + \sum_{j=1}^{N-1}\prod_{k=1}^{j}\gamma_k}.
\]
\end{lemma}

The following notation is helpful when stating results for large population sizes $N$.
\begin{defn}
    For functions $f,g$ of $N$ we write
\[
f(N)\approx g(N) \quad\text{if}\quad \lim_{N\to\infty}\frac{f(N)}{g(N)}=1,
\]

and 

\[f(N) \sim g(N)\]
if there exist constants positive constants $c_1, c_2$ such that 

$$ c_1 \leq \lim_{N \to \infty} \frac{f(N)}{g(N)} \leq c_2 .$$
\end{defn}

It is known that the fixation probability of a single oblivious mutant with fitness $r\ne 1$ on either the complete graph $K_N$ or on the cycle $C_N$ is equal to $(1-1/r)/(1-1/r^N)$. Thus, for $r<1$ we get the following asymptotics when $N$ is large.
\begin{lemma}
    Let $0< r < 1$, then we have 
    $$\fp_r(C_N) = \fp_r(K_n) \approx r^{N-1}(1-r) \sim r^N.$$
\end{lemma}

\begin{myproof}
    It is known that $\fp_r(C_N) = \fp_r(K_n) = \frac{1 -\frac1r}{1- \frac{1}{r^N}}$. Therefore 
    $$\fp_r(K_n) = \frac{1 -\frac1r}{1- \frac{1}{r^N}} = r^N \frac{1 - \frac1r}{r^N - 1 } \approx r^{N-1}\left(1 -r  \right) \sim r^N. $$
\end{myproof}

Finally, some of our proofs refer to the standard normal distribution.

\begin{defn}
    We denote by $\phi(x)$ and $\Phi(x)$ the density function and the cumulative distribution function of the standard normal distribution, respectively.
\end{defn}

\section{Neutral evolution}\label{sec:neutral}

In this section we compute the exact formulas for fixation probability and elimination probability on complete graphs in the neutral case ($r=1$), and derive their asymptotic behaviors. 

We first compute the general formula for fixation probability of a single replacer on a complete graph $K_N$. We will use this lemma in the proof of \cref{thm:1} and later in \cref{sec:non-neutral} in the proof of \cref{thm:complete_fix_disadv}.
The formula involves a term $\mathbb{P}(\Pois(r(N-1)) \leq N-1)$ which is the probability that a random variable distributed according to a Poisson distribution with parameter $\lambda=r(N-1)$ attains a value which is at most $N-1$.

\begin{lemma}\label{lem:sn_complete_closed}
    Let $r>0$. For a complete graph $K_N$ we have 
    $$\fps_r(K_N) = \frac{r^{N-1}(N-1)^{N-1}}{(N-1)! \cdot e^{r(N-1)}\cdot  \mathbb{P}(\Pois(r(N-1)) \leq N-1)}. $$
\end{lemma}
\begin{myproof}
      The probability of gaining a mutant is the probability of choosing one of the $k$ mutants for reproduction and as mutants are replacers, they will always choose a resident. Therefore $p^+_k = \frac{k}{T_k}\cdot r$.

    The probability of losing a mutant is the probability of choosing one of the $N-k$ residents for reproduction and then choosing one of the $k$ mutants. The probability of choosing a resident for reproduction is $\frac{N-k}{T_k}$ and as residents are oblivious, they will choose a mutants with probability $\frac{k}{N-1}$. Hence, $p^-_k = \frac{N-k}{T_k}\cdot\frac{k}{N-1}$.

    Then $\gamma_k = \frac{p^-_k}{p^+_k} = \frac{N-k}{(N-1)\cdot r}$ and using \cref{lem:one-dimensional} we have 

    $$\fps_r(K_N)  = \frac{1}{1 + \sum_{j=1}^{N-1}\prod_{k=1}^{j} \frac{N-k}{r(N-1)}}  = 
    \frac{1}{1
     + \sum_{j=1}^{N-1} \frac{(N-1)!}{(N-1-j)! (N-1)^j r^j}}  = 
     \frac{1}{\sum_{j=0}^{N-1} \frac{(N-1)!}{(N-1-j)!(N-1)^j r^j}}  =$$
     $$ = \frac{1}{\frac{(N-1)!}{(r(N-1))^{N-1}}\sum_{j=0}^{N-1} \frac{(r(N-1))^{N-j-1}}{(N-1-j)!}},$$
    and by re-indexing $l = N- 1 - j$ we get that the expression above is equal to 

    $$\fps_r(K_N) = \frac{1}{\frac{(N-1)!}{(r(N-1))^{N-1}} \sum_{l=0}^{N-1} \frac{(r(N-1))^{l}}{l!}} =\frac{1}{\frac{(N-1)!}{(r(N-1))^{N-1}} e^{r(N-1)} \sum_{l=0}^{N-1} \frac{(r(N-1))^{l}}{l!}e^{-r(N-1)}} =$$
    $$ = \frac{r^{N-1}(N-1)^{N-1}}{(N-1)!\cdot e^{r(N-1)} \mathbb{P}(\Pois(r(N-1)) \leq N-1)}.$$
\end{myproof}

To understand the asymptotics of the exact formula in~\cref{lem:sn_complete_closed}, we use a standard lemma about the Poisson distribution. For convenience we also sketch a proof.

\begin{lemma}\label{lem:Pois_mean_asymp}
    For $\lambda \in \mathbb{N}$ we have $$\mathbb{P}(\Pois(\lambda)  \leq \lambda) = \frac12 + \mathcal{O}(\lambda^{-1/2}). $$
\end{lemma}
\begin{myproof}
    From~\cite[Eq.~(4.18)]{barndorff1990asymptotic}, as cited in~\cite{janssen2008gaussian}, we have that

    \[\mathbb{P}(A_\lambda \leq s) = \Phi(\beta) - \frac{\phi(\beta)(\beta^2   -1)}{6\sqrt{\lambda}} + \mathcal{O}\left(\frac{1}{\lambda}\right),\]
    where $A_\lambda$ is the sum of $\lambda$ Poisson random variables with mean $1$, $\beta$ is an arbitrary constant, and $s = \lambda + \beta \sqrt{\lambda}$. Let $\beta = 0$. Then $s = \lambda$ and $\Pois(\lambda)$ can be written as a sum of $\lambda$ Poisson random variables with mean $1$, therefore we have 
    $$\mathbb{P}(\Pois (\lambda) \leq \lambda) = \Phi(0) - \frac{\phi(0)(-1)}{6\sqrt{\lambda}} + \mathcal{O}\left(\frac{1}{\lambda}\right) = \frac{1}{2} + \frac{1}{6\sqrt{2\pi \lambda}} + \mathcal{O}\left(\frac{1}{\lambda}\right) = \frac{1}{2} +\mathcal{O}(\lambda^{-1/2}) . $$
\end{myproof}

We are now ready to prove~\cref{thm:1} from the main text.

\begin{proof}[Proof of~\cref{thm:1}]
    According to \cref{lem:sn_complete_closed}, the fixation probability for $r=1$ is given by
    $$\fps_{r=1}(K_N)  = \frac{(N-1)^{N-1}}{(N-1)!\cdot e^{N-1}\mathbb{P}(\Pois(N-1) \leq N-1)}.$$ 
    By applying \cref{lem:Pois_mean_asymp} and then Stirling's approximation $k!\approx \sqrt{2\pi k}\cdot(k/e)^k$, we obtain
    \begin{align*}
        \fps_{r=1}(K_N)  &= \frac{(N-1)^{N-1}}{(N-1)!\cdot e^{N-1}\mathbb{P}(\Pois(N-1) \leq N-1)}\\
        &= \frac{(N-1)^{N-1}}{(N-1)!\cdot e^{N-1}(\frac{1}{2} + \mathcal{O}(N^{-1/2}))}\\
        &\approx\frac{(N-1)^{N-1}e^{N-1}}{\sqrt{2\pi (N-1)} (N-1)^{N-1} e^{N-1}(\frac{1}{2} + \mathcal{O}(N^{-1/2}))}\\
        &\approx\frac{1}{\sqrt{2\pi (N-1)}  \cdot \frac{1}{2} } \approx \frac{1}{\sqrt{N-1}} \cdot \sqrt\frac{2}{\pi}\approx \sqrt{\frac{2}{\pi\cdot N}}.
    \end{align*}
\end{proof}

Next, we proceed with the elimination probability. As before, first we compute a general formula for the elimination probability of a single replacer on a complete graph $K_N$. We will use the lemma in the proof of \cref{thm:2} and also later in \cref{sec:non-neutral} in the proof of \cref{thm:elim_complete_r_small}.

\begin{lemma}\label{lem:NS_pois}
Let $r>0$, then we have

$$ \eps_r(K_N) = \frac{1}{e^{r(N-1)}\mathbb{P}(\Pois(r(N-1)) \leq N-1)}.$$
\end{lemma}
\begin{myproof}
    Note that the extinction probability of $N-1$ mutants is the probability that a single resident spreads across the whole graph. We denote the number of residents by $k$. Let $p_k^+$ and $p_k^-$ be the probabilities of gaining and losing a resident, respectively. 
    The probability of gaining a resident is the probability of choosing one of the residents for reproduction and then selecting a mutant, thus $p_k^+ = \frac{k}{T_{N-k}}\cdot \frac{N-k}{N-1}$. Similarly, the probability of losing a resident is the probability of choosing one of the mutants, thus $p_k^- = r\cdot \frac{N-k}{T_{N-k}}$. Hence, we can calculate the bias $\gamma_k = \frac{p_k^-}{p_k^+} = \frac{r(N-k)}{T_{N-k}}\cdot\frac{(N-1) \cdot T_{N-k}}{(N-i)\cdot k} = \frac{r(N-1)}{ k}$. Therefore, using \cref{lem:one-dimensional} we have 
    $$ \eps_r(K_N)  = \frac{1}{1 + \sum_{j=1}^{N-1}\prod_{k=1}^j \frac{r(N - 1)}{ k} }  = \frac{1}{1 + \sum_{j=1}^{N-1}\frac{(r(N-1))^j}{ j!}} = \frac{1}{e^{r(N-1)}\sum_{j=0}^{N-1} \frac{(r(N-1))^j}{j!}\cdot e^{-r(N-1)}} = $$ $$ =\frac{1}{e^{r(N-1)}\mathbb{P}(\Pois(r(N-1)) \leq N-1)}.$$
\end{myproof}

Now we are in position to prove \cref{thm:2} from the main text.

\begin{proof}[Proof of~\cref{thm:2}]
By \cref{lem:NS_pois} for $r=1$ we have $$ \eps_{r=1}(K_N)  =\frac{1}{e^{N-1}\mathbb{P}(\Pois(N-1) \leq N-1)}.$$ Hence using \cref{lem:Pois_mean_asymp} we get
$$\eps_{r=1}(K_N)  =\frac{1}{e^{N-1}\mathbb{P}(\Pois(N-1) \leq N-1)} = \frac{1}{e^{N-1} (\frac{1}{2} + \mathcal{O}((N-1)^{-1/2}))} \approx \frac{2}{ e^{N-1}}.$$ 
\end{proof}

\section{Non-neutral evolution}\label{sec:non-neutral}

In this section we derive asymptotics for fixation probability and elimination probability on complete graphs in the non-neutral case.
We consider the disadvantageous replacers and the advantageous replacers separately.

\subsection{Disadvantageous replacers}

We make use of the following lemma to bound the distribution function of Poisson distribution in proofs of \cref{thm:complete_fix_disadv} and \cref{thm:elim_complete_r_small}
\begin{lemma}[{\cite[Theorem~2]{short2013poisson}}]\label{lem:KL_divergence}
    Let $X \approx \Pois(\lambda)$. Then 

    $$\mathbb{P}(X \leq k ) > \Phi\left(\sign(k - \lambda) \sqrt{2D_{KL}(Q_- \mid \mid P)}\right),$$
    where $D_{KL}(Q_- \mid \mid P) $ is the Kullback–Leibler divergence of $Q_- = \Pois(k)$ from $P = \Pois(\lambda)$. As $Q_-$ and $P$ are Poisson distributions, we have 

    $$D_{KL}(Q_- \mid \mid P) = \lambda - k + k\log\frac{k}{\lambda}.$$
    
\end{lemma}

First, we derive the asymptotic expression for the fixation probability of a single replacer.

\begin{theorem}\label{thm:complete_fix_disadv}
    Let $0 < r <1$. Then for a complete graph $K_N$ we have
    $$\fps_r(K_N) \approx \frac{r^{N-1} e^{(1-r)(N-1)}}{\sqrt{2\pi(N-1)}}.$$
\end{theorem}
\begin{myproof}
    Using \cref{lem:sn_complete_closed} we have that for $0 < r<1$ the fixation probability is equal to 

    $$\fps_r(K_N) = \frac{r^{N-1}(N-1)^{N-1}}{(N-1)! \cdot e^{r(N-1)} \mathbb{P}(\Pois(r(N-1)) \leq N-1)}.$$
    We use \cref{lem:KL_divergence} to bound $\mathbb{P}(\Pois(r(N-1)) \leq N-1)$ from below by a function that tends to 1 as $N\to\infty$. We have 
    $$\mathbb{P}(\Pois(r(N-1)) \leq N-1) > \Phi(\sign(N-1 - r(N-1)) \sqrt{2D_{KL}(Q_- \mid\mid P)}), $$
    where 
    $$D_{KL}(Q_- \mid\mid P)= (N-1)r - (N-1) + (N-1)\log\left(\frac{N-1}{(N-1)r}\right) = (N-1)(r - 1 - \log r). $$

    Let $c_r = (r -1 + \log \frac{1}{r})$ be a constant only dependent on $r$. We show that $c_r >0$ for $0< r<1$.
    We have $c_1 = 0$ and $\frac{d}{dr} = 1 - \frac{1}{r} < 0$ for $0<r<1$, hence $c_r$ is decreasing for $0< r < 1$ and $c_r$ is continuous at $r=1$, therefore $c_r >0$ for $0< r < 1$. We have $\sign(N-1 - r(N-1)) =1 $, hence 
    $\mathbb{P}(\Pois(r(N-1)) \leq N-1) > \Phi\left(\sqrt{2(N-1)c_r}\right) .$  As $\Phi$ is a continuous distribution function and $\lim_{N\to \infty} (N-1)c_r = \infty$, we have that $$\lim_{N\to \infty} \Phi(\sqrt{2(N-1 )c_r} = 1 \leq \lim_{N \to \infty}\mathbb{P}(\Pois(r(N-1)) \leq N-1) \leq 1,$$
hence $\lim_{N \to \infty}\mathbb{P}(\Pois(r(N-1)) \leq N-1) =1.$ Therefore we have  
$$\fps_r(K_N) = \frac{r^{N-1}(N-1)^{N-1}}{(N-1)! \cdot e^{r(N-1)} \mathbb{P}(\Pois(r(N-1)) \leq N-1)} \approx \frac{r^{N-1}(N-1)^{N-1}}{(N-1)! \cdot e^{r(N-1)} }.$$
Using Stirling's approximation we get
$$\fps_r(K_N) \approx \frac{r^{N-1}(N-1)^{N-1}}{(N-1)! \cdot e^{r(N-1)} }\approx \frac{r^{N-1}(N-1)^{N-1}e^{N-1}}{\sqrt{2\pi (N-1)}(N-1)^{N-1}e^{r(N-1)}} \approx \frac{r^{N-1}e^{(1-r)(N-1)}}{\sqrt{2\pi(N-1)}}.$$
\end{myproof}

Next, we find the asymptotic expression for the elimination probability of a single replacer.
\begin{theorem}\label{thm:elim_complete_r_small}
     Let $0 < r < 1$, then we have 
    $$\eps_r(K_N) \approx \frac{1}{e^{r(N-1)}}.$$
\end{theorem}

\begin{myproof}
By \cref{lem:NS_pois}, for $r<1$ we have $$\eps_r(K_N) = \frac{1}{e^{r(N-1)}\mathbb{P}(\Pois(r(N-1)) \leq N-1)}.$$ We use \cref{lem:KL_divergence}  
to bound the probability $\mathbb{P}(\Pois(r(N-1) \leq N-1)$ from below by a function that tends to 1 as $N\to\infty$. We have $$\mathbb{P}(\Pois(r(N-1)) \leq N-1) > \Phi(\text{sign}(N-1 - r(N-1)) \sqrt{2D_{KL}(Q_- \mid \mid P)} ) ,$$
where $$D_{KL}(Q_- \mid \mid P) = r(N-1) - N-1 + (N-1)\log\left(\frac{N-1}{r(N-1)}\right) = (N-1)\left(r  - 1 + \log\frac{1}{r}\right).$$
Let $c_r = (\frac{1}{r}  - 1 + \log r)$ be a constant only dependent on $r$. We show that $c_r > 0$ for $r>1$.

Let $c_r = (r -1 + \log \frac{1}{r})$ be a constant only dependent on $r$. We show that $c_r >0$ for $0< r<1$. we have $c_1 = 0$ and $\frac{d}{dr} = 1 - \frac{1}{r} < 0$ for $0<r<1$, hence $c_r$ is decreasing for $0< r < 1$ and $c_r$ is continuous in $r=1$, therefore $c_r >0$ for $0< r < 1$. We have $\sign(N-1 - r(N-1)) =1 $, hence 
    $\mathbb{P}(\Pois(r(N-1)) \leq N-1) > \Phi\left(\sqrt{2(N-1)c_r}\right) .$  As $\Phi$ is a continuous distribution function and $\lim_{N\to \infty} (N-1)c_r = \infty$, we have that $$\lim_{N\to \infty} \Phi(\sqrt{2(N-1 )c_r}) = 1 \leq \lim_{N \to \infty}\mathbb{P}(\Pois(r(N-1)) \leq N-1) \leq 1,$$
hence $\lim_{N \to \infty}\mathbb{P}(\Pois(r(N-1)) \leq N-1) =1.$ 

We derive the asymptotic behavior of the fixation probability from this as 

$$\lim_{N\to \infty} \eps_r(K_N) =  \lim_{N\to \infty}\frac{e^{r(N-1)}}{e^{r(N-1)} \mathbb{P}(\Pois(r(N-1)) \leq N-1)} =1,$$
therefore $$\eps_r(K_N) = \frac{1}{e^{r(N-1)}\mathbb{P}(\Pois(r(N-1) \leq N-1)} \approx \frac{1}{e^{r(N-1)}}.$$
\end{myproof}

To find fitness threshold for which the fixation probability is greater than elimination probability, we use the following theorem. 

\begin{theorem}
    We have 
    \begin{align*}
    \fps_r(K_N) > \eps_r(K_N) &\iff r> \frac{\sqrt[N-1]{(N-1)!}}{N-1},\\
 \fps_r(K_N) = \eps_r(K_N) &\iff r= \frac{\sqrt[N-1]{(N-1)!}}{N-1},\\
  \fps_r(K_N) < \eps_r(K_N) &\iff r<\frac{\sqrt[N-1]{(N-1)!}}{N-1},
     \end{align*}
    and for large enough $N$ we have $\frac{\sqrt[N-1]{(N-1)!}}{N-1}\to \frac{1}{e} \doteq 0.37$
\end{theorem}
\begin{proof}
    Using \cref{lem:sn_complete_closed} and \cref{lem:NS_pois} we have: 
    \begin{equation*}
     \fps_r(K_N) >\eps_r(K_N)
    \end{equation*}
    \begin{align*}
         &\iff \frac{r^{N-1}(N-1)^{N-1}}{(N-1)! \cdot e^{r(N-1)}\cdot  \mathbb{P}(\Pois(r(N-1)) \leq N-1)} >\frac{1}{e^{r(N-1)}\mathbb{P}(\Pois(r(N-1)) \leq N-1)}\\
        &\iff \frac{r^{N-1}(N-1)^{N-1}}{(N-1)!} > 1 \\
        &\iff r > \frac{\sqrt[N-1]{(N-1)!}}{N-1}.
    \end{align*}
    We can prove the other cases analogously. 

    For the second part we use Stirling's approximation $k!\approx \sqrt{2\pi k}\cdot(k/e)^k$:

    $$\frac{\sqrt[N-1]{(N-1)!}}{N-1} \approx \frac{\sqrt[N-1]{\sqrt{2\pi(N-1)}(N-1)^{N-1}/ e^{N-1}}}{N-1} \approx \frac{\sqrt[2N-2]{2\pi(N-1)}}{e} \approx \frac{1}{e}.$$

\end{proof}

\subsection{Advantageuos replacers}

We first derive the asymptotic behavior of thefixation probability.

\begin{theorem}\label{thm:complete_fix_adv}
    Let $r > 1$. For a complete graph $K_N$ we have $$\fps_r(K_N) \approx 1 - \frac{1}{r}.$$
    
\end{theorem}

\begin{myproof}  
    The probability of gaining a mutant is the probability of choosing one of the $k$ mutants for reproduction and as mutants are replacers, they will choose a resident for sure. Therefore $p^+_k = \frac{k}{T_k}\cdot r$.

    The probability of losing a mutant is the probability of choosing one of the $N-k$ residents for reproduction and then choosing one of the $k$ mutants. The probability of choosing a resident for reproduction is $\frac{N-k}{T_k}$ and as residents are oblivious, they will choose a mutants with probability $\frac{k}{N-1}$. Hence, $p^-_k = \frac{N-k}{T_k}\cdot\frac{k}{N-1}$.

    Then $\gamma_k = \frac{p^-_k}{p^+_k} = \frac{(N-k)\cdot k}{(N-1)\cdot k\cdot r} =  (1 - \frac{k-1}{N-1})\cdot \frac{1}{r}$. We use \Cref{lem:one-dimensional} to compute the fixation probability.

    $$\fps_r(K_N) = \frac{1}{1 + \sum_{j=1}^{N-1}\prod_{k=1}^{j}\frac1r\left(1 - \frac{k-1}{N-1}\right)}.$$ 

    Note that $\gamma_N = (1- \frac{N-1}{N-1}) =0$ and hence for $j \geq N$ we have  $\prod_{k=1}^{j}\frac1r\left(1 - \frac{k-1}{N-1}\right) = 0$. Therefore 

    $$\sum_{j=1}^{N-1}\prod_{k=1}^{j}\frac1r\left(1 - \frac{k-1}{N-1}\right) = \sum_{j=1}^{\infty}\prod_{k=1}^{j}\frac1r\left(1 - \frac{k-1}{N-1}\right).$$

We use the Dominated Convergence Theorem. Let $f_N(j) = \frac{1}{r^j}\prod_{k=1}^{j}\left(1- \frac{k-1}{N-1}\right)$ and $f(j) = \frac{1}{r^j}$. We have that $f_N(j) \leq f(j)$ and $\lim_{N \to \infty} f_N(j) = f(j)$ for all $j$. By the dominated convergence theorem we have that \[
1 + \lim_{N \to \infty}\sum_{j=1}^{\infty}\frac{1}{r^j}\prod_{k=1}^{j}\left(1- \frac{k-1}{N-1}\right) 
=
1 +\sum_{j=1}^{\infty}\lim_{N \to \infty}\frac{1}{r^j}\prod_{k=1}^{j}\left(1- \frac{k-1}{N-1}\right) =1+ \sum_{j=1}^{\infty}\frac{1}{r^j}\]
\[
 = \frac{1}{1-\frac1r}.\] Hence $\lim_{N \to \infty}\fps_r(K_N) = 1 - \frac1r$.
\end{myproof}

We note that \cref{thm:complete_fix_adv} together with~\cref{thm:complete_fix_disadv} immediately imply~\cref{thm:3-fp-r} from the main text.

Next, we bound the elimination probability for advantageous replacers. 

\begin{theorem}\label{thm:elim_complete_r_big}
    Let $r>1$ and $N \geq 3$, then we have 

    $$\eps_r(K_N) < r^N \cdot \frac{r^2}{r+1}.$$
\end{theorem}

\begin{myproof}
    We have $$\eps_r(K_N) \leq \ep_r(K_N) = \fp_{1/r}(K_N) = \frac{1-r}{1-r^N} = \frac{r-1}{r^N- 1}.$$
    As $r>1$ and $N \geq 3$, we have $r^{N-2} > 1$ and therefore 

    $$\eps_r(K_N) < \frac{r-1}{r^N - r^{N-2}} = \frac{1}{r^N} \cdot \frac{r-1}{1-\frac{1}{r^2}}= \frac{1}{r^N}\cdot \frac{r^2( r-1)}{r^2- 1} = \frac{1}{r^N}\cdot \frac{r^2}{r+1}.$$
\end{myproof}

We note that \cref{thm:elim_complete_r_big} together with~\cref{thm:elim_complete_r_small} immediately imply~\cref{thm:4-ep-r} from the main text.

\section{Cycle spatial structure}\label{sec:cycle-structure}

In this section we prove \cref{thm:cycle_fixation} about fixation probability on cycle $C_N$ and \cref{thm:cycle_elimination}  about elimination probability on cycle $C_N$. Using these theorems we get the two exact formulas in \cref{thm:5-cycle}. In the rest of the section, we then derive the other four asymptotic statements of \cref{thm:5-cycle}. 

Note that when we run the Moran process on a cycle graph $C_N$, the mutants always form a contiguous block.
This motivates the following definition.

\begin{defn}
Consider Moran process on a graph $C_N$ and let $r>0$ be the fitness of the replacers. Then:
\begin{itemize}
    \item      $\fps_r(C_N,i)$ denotes the mutant fixation probability when initially there are $i$ contiguous mutants who are replacers (and the residents are oblivious).

    \item $\eps_r(C_N, i)$ denotes the mutant elimination probability when initially there are $i$ contiguous mutnats who are replacers (and residents are oblivious).
\end{itemize}
\end{defn}

We present a closed-form formula for cycle and for any $r>0$ and any $i$.

\begin{theorem}\label{thm:cycle_fixation}
    Let $r>0$. Then we have

$$
 \fps_r (C_N, i) =\begin{cases}
     \frac{2i -1}{2N- 1} \text{\ \ if $r = \frac{1}{2}$,}\\
     \\
     \frac{1  - \frac{2}{(2r)^i} + \frac{1}{2r}}{1  - \frac{2}{(2r)^N} + \frac{1}{2r}} \text{\ \ if $r\not = \frac{1}{2}$.}
 \end{cases}
$$
   
\end{theorem}
\begin{myproof}
     As mutants always occupy a contiguous block, we have $p_1^+ = \frac{r}{T_1}$ and for $k \geq 2$, we have $p_k^+ = \frac{2r}{T_{k}}$. Residents occupy a contiguous block as well. Hence, if $k = N-1$, then $p_{N-1}^- = \frac{1}{T_{N-1}}$. If there are at least two residents, the case where $k \leq N-2$, there are two active residents that will each choose a mutant with probability $\frac12$, therefore $p_k^- = \frac{2}{T_{k}}\cdot \frac{1}{2} = \frac{1}{T_{k}}$. Therefore $\gamma_1 = \frac{1}{r}$ and $\gamma_k = \frac{1}{2r}$ for $k \geq 2$. Thus, by~\cref{lem:one-dimensional} we get

$$
\fps_r (C_N, i)  = \frac{1+\sum_{j=1}^{i-1}\prod_{k=1}^{j}\gamma_k}{1+ \sum_{j=1}^{N-1}\prod_{k=1}^j\gamma_k }=
\frac{1 + 2\cdot \sum_{j=1}^{i-1} \frac{1}{(2r)^j}}{1 + 2\cdot \sum_{j=1}^{N-1} \frac{1}{(2r)^j}} = \frac{2\cdot \sum_{j=0}^{i-1} \frac{1}{(2r)^j} - 1}{2\cdot \sum_{j=0}^{N-1} \frac{1}{(2r)^j} - 1}.$$
Now we distinguish the cases $r \not= \frac12$ and $r = \frac12$.
If $r = \frac{1}{2}$, then we have
$$\fps_{r=\frac12} (C_N, i)= \frac{2\sum_{j=0}^{i-1}1^j -1}{2\sum_{j=0}^{N-1} 1^j - 1} = \frac{2i -1}{2N- 1}. $$

If $r\not= \frac{1}{2}$, then we have
$$ \fps_{r\ne \frac12} (C_N, i)  = \frac{2\cdot\frac{1-\frac{1}{(2r)^i}}{1 - \frac{1}{2r}} -1 }{2\cdot\frac{1-\frac{1}{(2r)^N}}{1 - \frac{1}{2r}} -1  } = \frac{1  - \frac{2}{(2r)^i} + \frac{1}{2r}}{1  - \frac{2}{(2r)^N} + \frac{1}{2r}}.
$$ 
\end{myproof}

For a single initial mutant ($i=1$) the expressions simplify

\begin{rem}
    Let $r>0$. Then we have 
\begin{align*}
        \fps_{r=\frac{1}{2}}(C_N) &= \frac{1}{2N- 1} \text{\ \ if $r = \frac{1}{2}$ and }\\
        \fps_r (C_N) &= \frac{1  - \frac{1}{2r}}{1  - \frac{2}{(2r)^N} + \frac{1}{2r}} \text{\ \ if $r\not = \frac{1}{2}$}.
    \end{align*}
\end{rem}

Analogous calculation can be performed for the elimination probability.
\begin{theorem}\label{thm:cycle_elimination}
    Let $r>0$, then we have 

\[
\eps_{r=\frac12}(C_N,i)=\begin{cases}
    \frac{N-i}{N-\frac12} \text{\ \ if $r = \frac12$,}\\
    \\
    \frac{1 - (2r)^{N-i}}{1 - 2^{N-2}r^{N-1} - 2^{N-1}r^N} \text{\ \ if $r\not= \frac12$.}
\end{cases} 
\]

\end{theorem}

\begin{myproof}
    The elimination probability of $i$ mutants in a contiguous block is the probability that $l: = N-i$ residents in a contiguous block spread across the whole graph. Note that throughout the process, both mutants and residents always form a contiguous block. The total fitness in a configuration with $k$ residents is equal to $T_{N-k}$. We have that the probability of gaining a resident for $k := 1$ is  equal to $p_1^+ = \frac{1}{T_{N-1}}$ and for $2 \leq k \leq N-1$ it is equal to $p_k^+ = \frac{2}{T_{N-k}} \cdot\frac12 = \frac{1}{T_{N-k}}$. The probability of losing a resident is equal to $p_k^{-} = \frac{2r}{T_{N-k}}$ for $1 \leq k \leq N-2$ and for $k := N-1$ it is equal to $p_{N-1}^{-} = \frac{r}{T_1}$. Hence for $k\leq N-2$ we have $\gamma_k = 2r$ and for $k := N-1$ we have $\gamma_{N-1} = r$. Using \cref{lem:one-dimensional} we get

    $$\eps_r (C_N, i) =  \frac{1+\sum_{j=1}^{l-1}\prod_{k=1}^{j}\gamma_k}{1+ \sum_{j=1}^{N-1}\prod_{k=1}^j\gamma_k } =  \frac{1+\sum_{j=1}^{l-1}(2r)^j}{1+ \sum_{j=1}^{N-2}(2r)^j + 2^{N-2}r^{N-1}  }= \frac{\sum_{j=0}^{l-1}(2r)^j}{\sum_{j=0}^{N-2}(2r)^j + 2^{N-2}r^{N-1}}. $$

    Now we distinguish the cases $r= \frac12$ and $r\ne \frac12$.
    If $r = \frac12$, then we have
    $$\eps_{r=\frac{1}{2}}(C_N, i) = \frac{\sum_{j=0}^{l-1}1^j}{\sum_{j=0}^{N-2}1^j + \frac12} = \frac{l}{N-1 + \frac{1}{2}} =\frac{N-i}{N- \frac{1}{2}}.$$

    If $r\not= \frac12$, then we have
    $$\eps_r(C_N, i)= \frac{\frac{1-(2r)^l}{1-2r}}{\frac{1-(2r)^{N-1}}{1-2r} +2^{N-2}r^{N-1}}= \frac{1-(2r)^l}{1- 2^{N-1}r^{N-1} + 2^{N-2}r^{N-1} - 2^{N-1}r^{N}} = \frac{1-(2r)^{N-i}}{1 - 2^{N-2}r^{N-1} - 2^{N-1}r^{N}}.$$
\end{myproof}

For $N-1$ initial mutants and a single initial resident, we can simplify the expressions

\begin{rem}
    Let $r>0$. Then we have
    \begin{align*}
    \eps_{r=\frac12}(C_N) &= \frac{1}{N-\frac12} \text{\ \ if $r = \frac12$ and}\\
    \eps_{r\ne \frac12} (C_n) &= \frac{1 - 2r}{1 - 2^{N-2}r^{N-1} - 2^{N-1}r^N} \text{\ \ if $r\not= \frac12$,}
\end{align*}

\end{rem}

We are now ready to derive the asymptotics in~\cref{thm:5-cycle} from the main text.

\begin{proof}[Proof of~\cref{thm:5-cycle}]
    We start with the fixation probability. Let $r > \frac{1}{2}$. Then from \cref{thm:cycle_fixation} we have that 

    $$\rho_r^R (C_N) = \frac{1-\frac{1}{2r}}{1- \frac{2}{(2r)^N}  + \frac{1}{2r}} \approx \frac{1-\frac{1}{2r}}{1+\frac{1}{2r}} \approx \frac{1 + \frac{1}{2r} - \frac{1}{r}}{1+ \frac{1}{2r}} \approx1 - \frac{\frac{1}{r}}{1+\frac{1}{2r}} \approx1- \frac{1}{r + \frac{1}{2}}.$$

    Let $r < \frac{1}{2}$, then again from \cref{thm:cycle_fixation} we have that 
    $$\rho_r^R (C_N) = \frac{1-\frac{1}{2r}}{1- \frac{2}{(2r)^N}  + \frac{1}{2r}} = (2r)^N \frac{1-\frac{1}{2r}}{(2r)^N - 2 + \frac{(2r)^N}{2r}} \approx(2r)^N \frac{1-2r}{-2} \sim (2r)^N.$$

    We now derive asymptotics for the elimination probability. Let $r < \frac{1}{2}$. Then from \cref{thm:cycle_elimination} we have 
    $$ \eps_r(C_N) =  \frac{1 - 2r}{1 - 2^{N-2}r^{N-1} - 2^{N-1}r^N} = \frac{1-2r}{1 - \frac{1}{2 } (2r)^{N-1}  - \frac{1}{2} (2r)^N} \approx 1-2r. $$

    Let $r > \frac{1}{2}$, then we have 

     $$ \eps_r(C_N) =  \frac{1 - 2r}{1 - 2^{N-2}r^{N-1} - 2^{N-1}r^N} = \frac{1}{(2r)^N}  \frac{1 -2 r}{\frac{1}{(2r)^N} - \frac{1}{4r} - \frac{1}{2} } \approx \frac{1}{(2r)^N} \frac{1-2r}{-\frac{1}{4r}  -\frac{1}{2}} \sim \frac{1}{(2r)^N}. $$
 \end{proof}

 \section{Matrix Games}
Here we describe the connection between the evolutionary dynamics of replacers and the evolutionary games with frequency-dependent selection.

An evolutionary game between two types of individuals $A$ and $B$ on a spatial structure $G_N$ is given by a $2\times 2$ payoff matrix
\[M=\bordermatrix{ & A & B \cr
      A & p_{A,A} & p_{A,B} \cr
      B & p_{B,A} & p_{B,B} }
\]
of real numbers called \textit{payoffs}.
Initially, each site of $G_N$ is occupied by an individual of type $A$ or $B$.
In each step, each individual plays a game with each of its neighbors in $G_N$ and receives a payoff from each such interaction as prescribed by the matrix. (For instance, a type-$A$ individual receives payoff $p_{A,B}$ from interaction with a type-$B$ individual.)
The fitness of each individual is then defined to be the sum of the payoffs obtained from all its interactions.
Then, one step of the standard Moran process is performed.
This might change the composition of the population. The fitness of each individual is then recomputed using the same payoff matrix, and the steps are repeated until one type takes over the whole structure.

Let $p(M,G_N)$ be the probability that a single individual of type $A$ reaches fixation in the evolutionary game with matrix $M$, when played on spatial structure $G_N$.

The following theorem formalizes a connection between Moran process with replacers and an evolutionary game with a specific matrix.
Apart from the two spatial structures $K_N$ and $C_N$, it also involves a spatial structure $K'_N$, where each individual is considered a neighbor all other individuals and also to itself.
     
 \begin{lemma}
Fix a population size $N$ and replacer fitness $r>0$. Let $M_1=\bordermatrix{ & R & O \cr
      R & r & r \cr
      O & 0 & 1 }
$
 and $M_2=\bordermatrix{ & R & O \cr
      R & 3r & r \cr
      O & 1 & 1 }.
$
Then
\[
\fps_r(K'_N) = p(M_1,K'_N) \quad\text{and}\quad 
\fps_r(C_N) = q(M_2,C_N).
\]
 \end{lemma}
 \begin{proof}
     In each of the two cases, we show that the backward biases in Moran process with replacers are equal to the backward biases in the evolutionary game with the given matrix.

     We start with $K_N'$. In Moran process with replacers we have $p^+_k = r\cdot\frac{k}{T_k}$ and $p^-_k = \frac{N-k}{T_k} \cdot \frac{k}{N}.$ So the backward bias satisfies $$\gamma_k = \frac{p_k^-}{p_k^+} = \frac{N-k}{T_k} \cdot \frac{k}{N} \cdot \frac{1}{r}\cdot \frac{T_k}{k}\cdot  = \frac{1}{r}\cdot \frac{N-k}{N}.$$

     In the evolutionary game, if there are $k$ type $R$ individuals, then each type $R$ individual receive a total payoff $k \cdot r$ from all the type $R$ individual, and total payoff $(N-k) \cdot r$  from all the type $O$ individuals, so the fitness of a single type $R$ individual is $k \cdot r+ (N-k) \cdot r = N\cdot r$.
     Each type $O$ individual receives no payoff from all the type $R$ individuals and receives total payoff $N-k$ from all the type $O$ individuals, so the fitness of a single type $O$ individual is $N-k$. Using this we can compute the transition probabilities $q_k^+ = \frac{N\cdot r \cdot k}{T_k'}\cdot \frac{N-k}{N} $ and $q_k^- = \frac{(N-k)\cdot (N-k)}{T_k'} \cdot \frac{k}{N},$ where $T_k'$ is the total fitness in a state with $k$ type $R$ individuals. Therefore the backward bias in this case is equal to $$\gamma_k' = \frac{q_k^-}{q_k^+} = \frac{(N-k)\cdot (N-k)}{T_k'}\cdot \frac{k}{N} \cdot \frac{T_k'}{N\cdot r\cdot k} \cdot \frac{N}{N-k}=\frac{1}{r}\cdot \frac{N-k}{N} = \gamma_k.$$

    We now show it for $C_N$. In Moran process with replacers we have $\gamma_1 = \frac{1}{r}$ and for $k \geq 2$ we have $\gamma_k = \frac{1}{2r}$.
    In the evolutionary game, both type $R$ individuals and type $O$ individuals form a contiguous block. To find a backward bias, it is enough to focus on individuals at the boundaries of the blocks. If there is a single type $R$ individual, he receives total payoff $2\cdot r$ from the neighboring type $O$ individuals. If there are at least two type $R$ individuals in the block, then each of the type $R$ individual at the boundary receives payoff $3r + r = 4r$ from its neighbors. Therefore $q_1^+ = \frac{2r}{T_1}$ and for $k\geq 2$ we have $q_k^+ = \frac{2\cdot 4r}{T_k'}\cdot \frac{1}{2} = \frac{4r}{T_k'}$. 

    Now we focus on the type $O$ individuals at boundaries. As all the type $O$ individuals get the same payoff from type $R$ and type $O$, all the type $O$ individuals have fitness $2\cdot 1$. Thus for  $k\leq N-2$ we have $q_k^- = \frac{2\cdot 2}{T_k'}\cdot \frac{1}{2}= \frac{2}{T_k} = q_{N-1}^-$. So we have, that the backward bias when there is a single type $R$ individual is equal to $$\gamma_1' = \frac{q_1^-}{q_1^+} = \frac{2}{T_1'} \cdot \frac{T_k'}{2r} = \frac{1}{r} = \gamma_1, $$
    and for $k\geq 2$ we have

    $$\gamma_k' = \frac{q_k^-}{q_k^+} = \frac{2}{T'_k} \cdot \frac{T_k'}{4r} = \frac{1}{2r} = \gamma_k,$$
    as we wanted to show.
     
 \end{proof}

\end{document}